%% file: OverComBackscatter_arXiv.tex
\documentclass[12pt, draftclsnofoot, onecolumn]{IEEEtran}	% regular single-column
\renewcommand{\baselinestretch}{1.4}

\usepackage{fancyhdr, epsfig, epsf, amsthm, amsmath, amssymb, amsfonts, subfigure, color}
\usepackage{threeparttable}
\usepackage[noadjust]{cite}
\usepackage{dsfont}
\usepackage{enumerate}
\usepackage{comment}
\usepackage{soul}
\usepackage{algorithm}
\usepackage{algpseudocode}
\makeatletter
\newenvironment{subtheorem}[1]{%
  \def\subtheoremcounter{#1}%
  \refstepcounter{#1}%
  \protected@edef\theparentnumber{\csname the#1\endcsname}%
  \setcounter{parentnumber}{\value{#1}}%
  \setcounter{#1}{0}%
  \expandafter\def\csname the#1\endcsname{\theparentnumber.\Alph{#1}}%
  \ignorespaces
}{%
  \setcounter{\subtheoremcounter}{\value{parentnumber}}%
  \ignorespacesafterend
}
\makeatother
\newcounter{parentnumber}

\newtheorem{thm}{Proposition}
\include{header}

\includecomment{comment}	% figure on/off

\def\papertitle{\Huge Inference from Randomized Transmissions by Many Backscatter Sensors}

\IEEEoverridecommandlockouts 

\begin{document}

\title{ \fontsize{21}{21}\selectfont \papertitle}
\author{Guangxu Zhu, Seung-Woo Ko and Kaibin Huang
%\thanks{\IEEEauthorrefmark{1}\hl{These authors contributed equally to this work.}}
\thanks{ G. Zhu, S.-W.~Ko and K.~Huang are with the Department of Electrical and Electronic Engineering, The University of Hong Kong, Pok Fu Lam, Hong Kong (e-mail: gxzhu@eee.hku.hk, swko@eee.hku.hk, haungkb@eee.hku.hk). %Updated on \today
}}
\maketitle

\vspace{-15mm}

\begin{abstract}
Attaining the vision of Smart Cities requires the deployment of an enormous number of sensors for monitoring various conditions of the environment ranging from air quality to traffic. Backscatter sensors have emerged to be a promising solution for two reasons. First, transmissions by backscattering allow sensors to be powered wirelessly by radio-frequency (RF) waves, overcoming the difficulty in battery recharging for billions of sensors. Second, the simple backscatter hardware leads to low-cost sensors suitable for large-scale deployment. On the other hand, backscatter sensors with limited signal-processing capabilities are unable to support conventional algorithms for multiple access and channel training. Thus, the key challenge in designing backscatter sensor networks is to enable readers to accurately detect sensing values given simple ALOHA random access, primitive transmission schemes, and no knowledge of channel states and statistics. We tackle this challenge by proposing the novel framework of \emph{backscatter sensing} (BackSense) featuring random encoding at backscatter sensors and statistical inference at readers. Specifically, assuming the widely used on/off keying for backscatter transmissions, the practical random-encoding scheme causes the on/off transmission of a sensor to be randomized and follow a distribution parameterized by the sensing values. Facilitated by the scheme, statistical inference algorithms are designed to enable a reader to infer sensing values from randomized transmissions by multiple backscatter sensors. The specific design procedure involves the construction of \emph{Bayesian networks}, namely deriving conditional distributions for relating unknown parameters and variables (including sensing values, noise power, sensing measurements, number of active sensors) to signals observed by the reader. Then based on the Bayesian networks and the well-known \emph{expectation-maximization} (EM) principle, inference algorithms are derived to recover sensing values. Simulation of the BackSense system demonstrates high accuracy in reader inference despite the mentioned limitations of backscatter sensors, which grows with increasing numbers of received symbols and reader antennas. 
\end{abstract}

%\begin{IEEEkeywords}
%Wireless sensor network, backscatter communication, statistical inference, Bayesian network.
%\end{IEEEkeywords}

\vspace{-3mm}
\section{Introduction}\label{Section:Introduction}

Realizing the visions of \emph{Internet-of-Things} (IoT) and Smart Cities will require the deployment of billions of wireless IoT sensors in our society for automating a wide range of applications such as health care, smart  homes, pollution and traffic monitoring, and industrial control. Backscatter transmission has emerged to be a promising solution for tackling some key challenges on deploying  large-scale sensor networks. The main feature of a backscatter radio  is to transmit by  backscattering  and modulating   an incident \emph{radio-frequency} (RF) wave \cite{Boyer2014}. This allows backscatter sensors to be powered wirelessly by a reader, overcoming the difficulty of  battery recharging for a massive number of sensors. Furthermore, requiring no  oscillators and RF components, backscatter sensors can be manufactured to have a small form factor and low cost. Thereby, they are  suitable for  large-scale deployment.  Last, with recent advancements in multi-antenna beamforming and low-power electronics, the ranges for backscatter links have been increased from several meters in the classic RFID applications to tens of meters in state-of-the-art systems \cite{Yang2015, Kimionis2014}. This makes it possible to collect measurement data from backscatter sensors using versatile  mobile readers mounted on  vehicles and  \emph{unmanned aerial vehicles} (UAVs) (see Fig.~\ref{Fig:UAV_sensing}). 
Motivated by the potential of backscatter sensor networks, this work presents  a novel  design framework, called randomized \emph{backscatter sensing} (BackSense),  for sensing-data uploading and detection based on random access and  machine learning. To be specific, to resolve transmission collisions, statistical-inference algorithms are designed for  readers to directly infer sensing values from collided signals by exploiting their spatial correlation and partial information of their distributions. The algorithms build on  a proposed scheme of randomized on/off transmissions targeting  backscatter sensors. 

\begin{figure}[t]
\centering
\includegraphics[width=12cm]{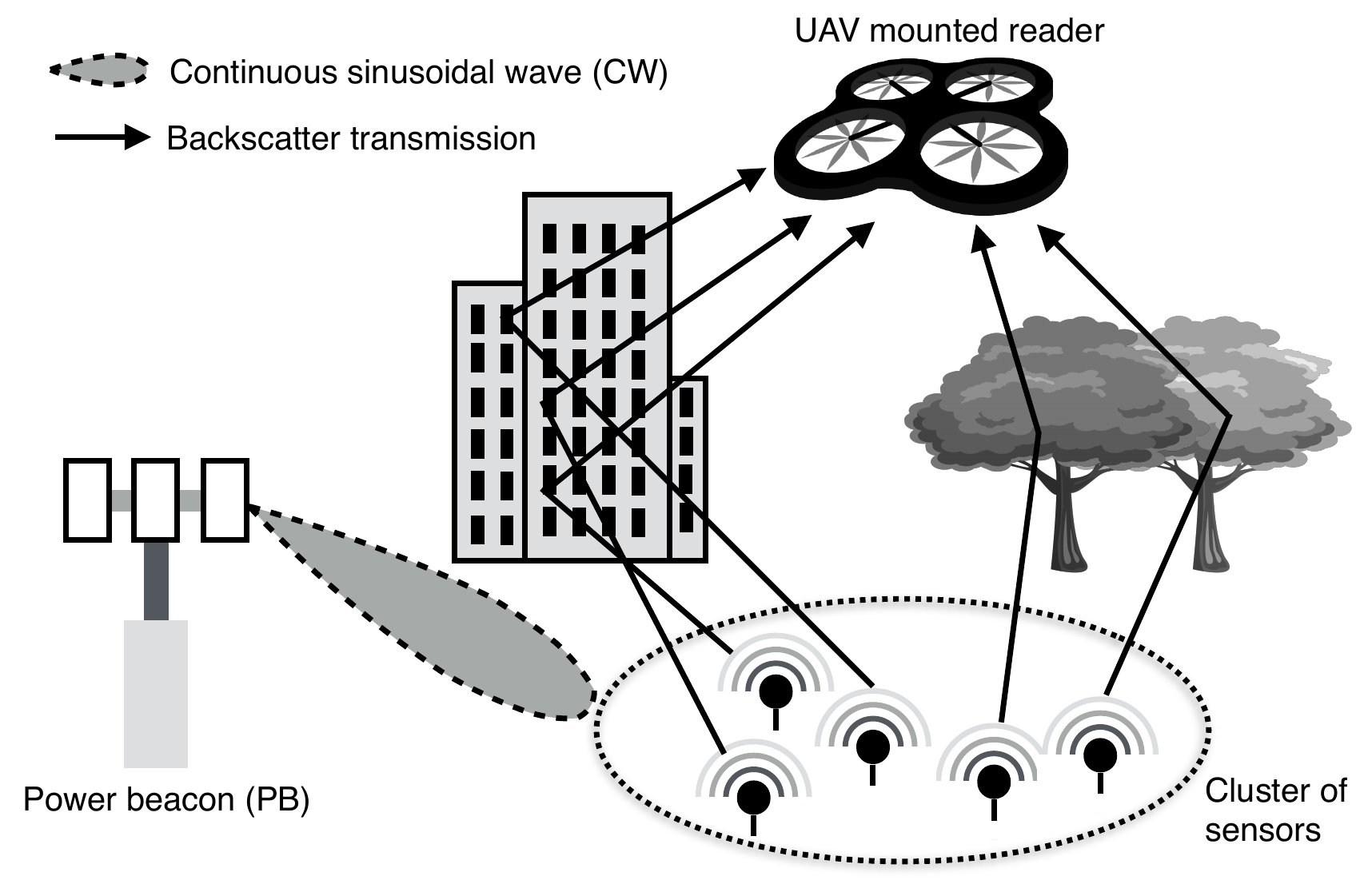}
\caption{A BackSense system comprising  a cluster of sensors deployed at a remote or hazardous location, a PB wirelessly powering the sensors, and a UAV mounted reader for collecting sensing data.}
\vspace{-6mm}
\label{Fig:UAV_sensing}
\end{figure}

\subsection{Backscatter Communications and Networking}
Backscatter communications are expected to play a key role in next-generation sensor systems  e.g., IoT and Smart Cities. These applications motivate   growing research on developing advanced techniques for backscatter communications that are more complex than those for the traditional RFID applications \cite{Boyer2014, Clerckx2017, Qian2017, Alevizos2017, Liu2017,Han2017}. The modulation and space-time coding schemes were designed in \cite{Boyer2014} for enhancing the rate and reliability of backscatter communication links. Traditionally, a reader transmits a  single sinusoidal wave  for powering a backscatter tag but such a design is sub-optimal in terms of energy harvesting efficiency. Instead, a waveform design superimposing multiple sinusoidal waves was  proposed  in  \cite{Clerckx2017} for enhancing the efficiency, via optimization  based on the harvester-response function. For backscatter tags with low-complexity, non-coherent detection schemes are preferred due to the simple transmitter architecture. Specific schemes based on on/off keying or frequency-shift keying were  proposed in \cite{Qian2017} and \cite{Alevizos2017},  respectively.     
The channel-utilization efficiency for  backscatter communications  tends to be low as the primitive backscatter architecture is incapable of supporting sophisticated transmission designs such as high-order modulation or spatial multiplexing. An attempt was  made in \cite{Liu2017} on improving the efficiency by designing full-duplex backscatter communication networks. In this design,  full-duplexing is realized by \emph{double} (coherent and non-coherent) modulating  backscattered signals and using time-hopping spread spectrum to mitigate interference. Last, to extend the ranges as well as simultaneoully power a large number of sensors, a novel architecture for backscatter communication networks was recently proposed in  \cite{Han2017} where tags are powered by distributed \emph{power beacons} (PBs) instead of readers. Then the tradeoff between network coverage and  PB density was studied therein  using a stochastic-geometric network model. 

A key design issue for backscatter sensor networks is multiple access by sensors. The natural design approach is to adopt  classic multiple-access schemes including  random access \cite{Zhen2005, Isik2009}, \emph{time division multiple access} (TDMA) \cite{Zheng2016, Hoang2017}, \emph{space-division multiple access} (SDMA) \cite{Angerer2010,Kim2011}, and \emph{code division multiple access} (CDMA)  \cite{Demeechai2011,Wang2012}. As mentioned, for large-scale sensors networks, these schemes have   the drawbacks of excessive network overhead due to protocols and algorithms for orthogonalizing  transmissions   as well as long latency and low spectrum-utilization  efficiency. Moreover, the required signal processing (e.g.,  FFT,  channel estimation and feedback)  are impractical for backscatter  sensors with a primitive   hardware architecture. On the other hand, ALOHA-type random access does not have the above drawbacks but suffers from performance degradation caused by  transmission collisions. A novel approach for  resolving collisions is proposed  in \cite{Wang2012}. Specifically, assuming sparse random access by backscatter tags, collided signals are treated as a sparse code and recovered using the techniques of compressive sensing and rateless coding. The design is effective only in the scenario of dense but mostly inactive sensors and thus cannot support the deployment of dense sensor  networks with many active sensors. 

\subsection{Machine Learning for Wireless Sensor Networks} 
Machine learning provides  a rich set of techniques for learning and prediction of data. Recent breakthroughs in the field especially the area of artificial intelligence motivated researchers to apply relevant techniques for bringing intelligence to communication systems and networks   (e.g., UAV assisted networking  \cite{Chen2017_ICC2} and molecular communication \cite{farsad2017detection}). In particular, a thrust of the research focuses on applying machine-learning techniques  to streamline  various  operations of sensors networks such as routing  \cite{Barbancho2007}, abnormality  detection \cite{Krishnamachari2004}, positioning and  network topology discovery \cite{Morelande2008}, and sensing-data compression and  feature extraction \cite{Kho2009}. 

The existing research most relevant to the current work is the application of statistical-learning techniques to detection of sensing values  from signals transmitted by wireless sensors based on random access \cite{Mergen2006,Chen2006,Glaropoulos2015, liu2006rlmac, shen2008broadcast}.  In  \cite{Mergen2006,Chen2006}, the sensors are designed such that their noisy measurements received  by the reader provide sufficient statistics of the measured environmental parameter, thereby allowing it to be inferred using the classic maximum-likelihood algorithm. This requires each sensor to estimate and inverse its channel to the reader and furthermore map the noisy measurement to a pre-determined set of orthogonal waveforms. Such operations, however, are too complex for a backscatter sensor that usually supports only  primitive signal processing and modulates signals using the simple on/off keying. In a series of other  research,  advanced machine-learning techniques were applied to  enhance the performance of the  \emph{carrier-sensing multiple access} (CSMA) scheme \cite{Glaropoulos2015, liu2006rlmac, shen2008broadcast}.  In \cite{Glaropoulos2015}, cognitive transmitting sensors  rely on statistical inference to infer the channel status (idle or busy) at the receiver based on local measurements. Such information allows  the sensors to intelligently  switch between active and sleep modes  to avoid collision and reduce energy consumption. A different kind of cognitive sensor is designed in  \cite{liu2006rlmac} based on reinforcement learning to have the capability of adapting its  duty cycle according to the traffic load and channel state, thereby reducing its energy consumption as well as enhancing  network throughput.  Another  technique, namely neural network, is applied in  \cite{shen2008broadcast} to learn the optimal  scheduling policy that minimizes a schedule cycle such that all sensors can be scheduled within each cycle without any collision. The  existing designs of cognitive transmitters are impractical for backscatter sensors again due to their low complexity and limited signal-processing capability. The  more practical approach, especially for large-scale sensor networks, should be one reducing sensor complexity and compensating for it by centralized  machine learning at readers or in the cloud. 

\subsection{Contributions and Organization}
Based on this approach, we design techniques for efficient implementation of backscatter sensor networks. Specifically, adopting the network architecture proposed in \cite{Han2017}, we consider a system illustrated in Fig. \ref{Fig:UAV_sensing}  where a cluster of backscatter sensors are powered wirelessly by a PB for transmission to a mobile reader  (also called an information collector) mounted on e.g., a UAV. The BackSense  framework for designing such a  system satisfies  several practical constraints. First, the transmission of each sensor is  based on on/off keying.  Second, multiple access by sensors is based on the simple ALOHA-type  random access (without carrier sensing and scheduling). Last, both sensors and reader  have neither   \emph{channel state information} (CSI) nor parameters of channel distribution. The only knowledge the reader has about the channels is their distribution type. Under these constraints, the BackSense framework is designed for efficient sensing-data uploading and  accurate detection. To the best of the authors' knowledge, this current work represents the first attempt on applying  machine learning  to the design of backscatter-sensor   systems. The resultant BackSense framework  comprises two key  components: one is \emph{randomized transmissions} by sensors and the other \emph{statistical inference} at the reader, described as follows.

{\bf Randomized transmission for a backscatter sensor:} Under the constraint of on/off keying, a novel random-encoding scheme for sensor transmission   is proposed for embedding the sensing value into the distribution of backscattered signals.  Mathematically, the sensor state (backscatter or not)  is distributed as a Bernoulli \emph{random variable} (r.v.) whose distribution is parameterized by the sensing value. In other words, the sensing value governs the transmission probability. The randomized transmission  is realized by the proposed encoder design that one-to-one maps the sensing value using a sigmoid function (one with a ``S" shape) to yield a variable with  a normalized range and then comparing the result with a uniform r.v. in the same  range, generating a sequence of binary bits. The bits switch the connection of the sensor antenna with either of two load impedances and as the result, randomly turn backscatter transmission on/off, thereby modulating the bits by on/off keying. 

{\bf Statistical inference at the reader:} Statistical inference is capable of estimating parameters of a signal distribution based on signal observation. Using the theory and given randomized transmission by sensors, the reader is designed to infer their transmission probabilities (or equivalently their sensing values) from observations over time and antennas. To this end, algorithms for reader inference  are designed based on the \emph{expectation-maximization} (EM) framework. A simpler version of the algorithms are  designed assuming no measurement noise and the assumption is relaxed subsequently. To design the algorithms, a Bayesian network is constructed that relates the signals received by the reader to the sensing values via a number of fixed or random variables including the mapped sensing values, number of active sensors, variances of channel noise and channel gains. Their relations are specified by their  conditional distributions as derived. Building on the Bayesian network, the specific EM-based  algorithms targeting the BackSense system are developed by deriving two iterative steps, called  the Expectation and Maximization steps, for both the criteria of \emph{maximum likelihood} (ML) and \emph{maximum a posteriori} (MAP). 

The inference algorithms are  extended to the case with measurement noise. As a result,  the Bayesian network in the preceding  case should be modified to include a set of new r.v., modelling measurement noise, as many as  the number of  sensors. This  dramatically increases the dimensionality of the \emph{latent-space}, the space of latent variables (r.v. unobserved by the reader). As the result, the  corresponding EM-based algorithms become  computationally demanding and  thus impractical. To overcome the difficulty, we exploit an approximate implementation of the EM-framework based on the \emph{variational-inference} method \cite{beal2003variational}. Its basic idea for complexity reduction is to  restrict the approximate posterior distribution of the latent variables to take a factorized form. This converts the required a larger number of nested  integrals over the latent-variable  space in the EM framework to  parallel integrals, enabling computationally-efficient implementation. Based on the method, we develop the EM-based  algorithms for inference at the reader for the ML and MAP criteria. 

The remainder of the paper is organized as follows.  Section II introduces the system model and  problem formulation. Section III presents the randomized transmission scheme at tags. The EM-based algorithms for reader inference are designed  in Sections IV and V for the cases without and with  measurement noise,  respectively.  Simulation results are provided in Section VI, followed by concluding remarks in Section VII.

\section{System Model and Problem Formulation}\label{Section:SystemModel} 

%\vspace{-3mm}
\subsection{System Model}
We consider a sensing  system (see Fig. \ref{Fig:Architecture}) comprising a PB, a reader equipped with $M$ antennas, and $N$ single-antenna backscatter sensors. Provisioned with reliable power supply, the PB is equipped with an antenna array and able to power the sensors by energy beamforming. The reader uses the antennas for receiving simultaneous signals transmitted by the sensors.  In the following, the models of backscatter sensor, wireless channels  and sensing values are described. 

Provisioned with a backscatter antenna, each sensor reflects a fraction of \emph{continuous wave} (CW) back to the reader and harvests the energy of the remaining fraction for powering the sensor circuit.  In the process, the sensor modulates the reflected CW as follows.  First of all, note  that the  reflection coefficient~$S$, the ratio between incident and reflected CWs, depends on the level of mismatch between the antenna and load impedances. Mathematically,  $S=\frac{Z_L-{Z_A}^*}{Z_L+{Z_A}^*}$, where $Z_L$ and $Z_A$ are load and antenna impedances, respectively \cite{Boyer2014}. The sensor modulates the backscattered CW by adapting the reflection coefficient that determines the phase and magnitude of the wave \cite{Boyer2014}.  The coefficient variation  can be implemented  by switching over a set of load impedances (see Fig. \ref{Fig:Architecture}). We consider the  modulation scheme based on on/off keying that is widely used for backscatter transmission for its simplicity. The scheme requires switching between two load impedances $Z_1 $ and $Z_2$ with $Z_1 \neq Z_A^*$ and $Z_2 = Z_A^*$. In other words, for a sensor, switching to $Z_1$ turn on  backscatter transmission ($S  = \frac{Z_1-{Z_A}^*}{Z_1+{Z_A}^*} = \bar \rho$) and to $Z_2$ turns transmission off ($S = 0$). 

The channel model is described as follows.  Two kinds of channels are cascaded. One is from the PB to the sensors for \emph{energy transfer} (ET) and 
the other is from the sensors to the reader for \emph{information transfer} (IT).  
First, given energy beamforming and assuming that the size of the sensor cluster is much  smaller than its distance to the PB, ET channels can be modelled as \emph{line-of-sight} (LOS) channels with identical path losses. Second, with longer distances and isotropic wave propagation, IT channels  are assumed to be characterized by rich scattering and hence modelled  as Rayleigh fading. Specifically, time is slotted into symbol durations and we consider block fading  where channels coefficients are  \emph{independent and identically distributed} (i.i.d.) over  different time slots. The IT channel  in slot $i$ is denoted as a $M \times N$ matrix $\mathbf{H}^{(i)}$,  whose $(\ell,n)^{\mathrm{th}}$ element, denoted by $h_{\ell,n}^{(i)}$, represents the coefficient from sensor $n$ to the reader's antenna $\ell$. 
It is assumed that all channel  coefficients $\{h_{\ell,n}^{(i)}\}$ are i.i.d.  $\mathcal{CN}(0, \sigma_h^2)$ r.v.. 
Let $S_n^{(i)}$ be the reflection coefficient of sensor  $n$ in slot  $i$. 
The vector collecting the reflection coefficients for all $N$ sensors in slot $i$ is then denoted by
$\bs^{(i)}=[S_1^{(i)}, S_2^{(i)}, \cdots, S_N^{(i)}]^{\mathsf{T}}$. 
Given $\bs^{(i)}$, the received signal at the reader in slot $i$, denoted by $\mathbf{y}^{(i)}=[y_1^{(i)}, y_2^{(i)}, \cdots, y_M^{(i)}]^{\mathsf{T}}$, is as follows\footnote{The interference from PB to the reader can be easily canceled as it is a CW and thus is not considered.}
\begin{align}\label{Eq:Signalmodel}
\mathbf{y}^{(i)}=\sqrt{G P_t}\mathbf{H}^{(i)}\bs^{(i)}+\mathbf{w}^{(i)}, \qquad i=1, \cdots, L,  
\end{align}
where $\mathbf{w}^{(i)}=[w_1^{(i)}, w_2^{(i)}, \cdots, w_M^{(i)}]^{\mathsf{T}}$ is the \emph{additive white Gaussian noise} (AWGN) with the entries following i.i.d. $\mathcal{CN}(0, \sigma_w^2)$ distributions, $L$ denotes the length of the observation  period in slot, $G$ captures the energy beamforming gain, and $P_t$ is the transmit power of  the PB. Without loss of generality, $G$ and $P_t$ are set as one to simplify notation.

\begin{assumption}[CSI Free] \label{Asumption:NoCSI}\emph{No CSI of individual
backscatter channels is  available at both of the reader and the sensors.  The reader only has the knowledge of distribution type  of  ${\bf H}^{(i)}$ and ${\bf w}^{(i)}$ (i.e., complex Gaussians),  but not the distribution parameters $\sigma_h^2$ and $\sigma_w^2$.}
\end{assumption}

The sensing values measured by different sensors are spatially correlated (e.g., temperature, humidity, and wind strength). The model of their joint distribution is described as follows. Let the sensing values to be measured by $N$ sensors be denoted by $\bx = [x_1, x_2, \cdots, x_N]^{\mathsf{T}}$. It is assumed that the prior  knowledge of the distribution of $\bx$ is available at reader  by estimation using historical measurement data.   For tractability, $\bx$ is assumed to follow the multivariate Gaussian distribution  $\bx \sim {\cal N}(\boldsymbol \mu_\bx, \boldsymbol \Sigma_\bx)$ with mean $\boldsymbol \mu_\bx$ and covariance matrix $\boldsymbol \Sigma_\bx$. The \emph{probability density function} (PDF) of $\bx$ is then given by
\begin{align}\label{PDF of x}
p(\bx) = \frac{1}{(2\pi)^{\frac{N}{2}} |\boldsymbol \Sigma_\bx|^{\frac{1}{2}}}\exp\l( -\frac{1}{2} (\bx - \boldsymbol \mu_\bx)^T \boldsymbol \Sigma_\bx^{-1} (\bx - \boldsymbol \mu_\bx)\r).
\end{align}

\begin{assumption}[Time-Scale Differentiation]
\emph{The environmental features usually vary  much slower than wireless channels especially given reader mobility. Thereby,  the sensing values $\mathbf{x}$ are assumed to remain unchanged during the whole observation window  of  $L$ time slots, while wireless channels are i.i.d. over slots.} 
\end{assumption}

In addition, the sensor measurement of  a sensing value $x$, denoted as $\tilde{x}$, is corrupted by measurement noise: $\tilde{x} = x + \Delta$ where  the r.v. $ \Delta$ represents the noise and assumed to follow  the $\mathcal{N}(0, \delta^2)$ distribution. Moreover, measurement noises at different sensors and slots are  i.i.d..

\subsection{Problem Formulation}
The conventional \emph{deterministic} transmission design that first quantizes the sensing value and then transmits the output bits is unsuitable for a backscatter sensor without an analog-to-digital converter and under the constraint of on/off keying. To overcome the limitations, the proposed BackSense design framework relies on a randomized design for sensor transmission  and matching statistical-inference algorithms for the reader to recover the transmitted sensing values. Two corresponding design problems are formulated as follows. 

\subsubsection{Randomized Transmission  Problem for Backscatter Sensors}\label{randomized transmission problem}
Consider an arbitrary symbol duration. Let $S_n$ with $S_n \in \{0, \bar{\rho}\}$ denote the symbol transmitted by sensor $n$. The symbols from $N$ sensors are grouped as a vector $\bs = [S_1, S_2, \cdots, S_N]^T$. The design objective is to randomize sensor transmissions such that the sensing value vector $\bx$ is encoded into 
a sequence of realizations of the random vector $\{\bs^{(i)}\}$, whose distribution function, $p(\bs^{(i)})$,  is parameterized by the sensing values in $\bx$.  Designing the mentioned random-encoding scheme is challenging. First, it should be implementable under the constraint of on/off keying and using the simple backscatter hardware architecture. Next, the scheme should yield a distribution function $p(\bs^{(i)})$ that allows tractable design of statistical-inference algorithms for the reader. 

\subsubsection{Statistical Inference Problem for Reader} \label{inference problem}
The algorithms for statistical inference at the reader are designed targeting randomized transmissions by backscatter sensors from solving  the preceding design problem.  Furthermore, under the constraints of ALOHA-type random access and no CSI, the algorithms should be capable of scaling  up the inference accuracy with growing correlation between sensing values and the number of received observations. To this end, we design the inference algorithms based on the ML or MAP criterion. Let $\by^{(i)}$ denote the symbol vector received by the multi-antenna reader in the slot $i$. The set of data received over  $L$ slots  is represented by $\cal{D}$. Moreover, let  $\boldsymbol \Phi = \{\bx, \sigma_h^2, \sigma_w^2, \delta^2\}$ denote  the set of all the unknown parameters including the desired sensing value $\bx$. Then the ML problem formulation is given as follows: 
\begin{align}({\bf ML})\qquad 
\max_{\boldsymbol \Phi} \log p({\cal D}| \boldsymbol \Phi) \iff \max_{\boldsymbol \Phi} \sum_{i=1}^L \log p(\by^{(i)}|\boldsymbol \Phi). \label{ML formulation}
\end{align}
If the prior distribution of $\boldsymbol \Phi$ is available, the MAP  algorithms can be obtained using  the following MAP-to-ML conversion: 
\begin{align}({\bf MAP})\qquad
\max_{\boldsymbol \Phi} \log p(\bx|{\cal D}) &\iff \max_{\boldsymbol \Phi} \log \frac{p({\cal D}|\boldsymbol \Phi) p(\boldsymbol \Phi)}{p(\cal D)}\notag\\ 
&\iff \max_{\boldsymbol \Phi} \sum_{i=1}^L \log p(\by^{(i)}|\boldsymbol \Phi) + \log  p(\boldsymbol \Phi).\label{MAP formulation}
\end{align}
The challenge  for designing the inference algorithms based on the ML or MAP criteria arises from the existence of latent variables. In statistics, a \emph{latent variable} is defined as a variable not directly observed but has a direct or indirect impact on the observed variables. Denote $\bz$ as the vector that collects all the involved latent variables, which capture the randomized effect of the measurement noise, random transmitted symbols at  sensors, and channel noise and coefficients as specified in the sequel. %(see Fig. \ref{Fig:Probabilistic model without noise} and \ref{Fig:Probabilistic model with noise}). 
Note that, via the latent variables, the design of random-encoding scheme in the preceding problem formulation affects the design of inference algorithms. 
%The number of latent variables increase linearly with the number of sensors, causing  the inference complexity to grow rapidly. 
The difficulty of solving \eqref{ML formulation} or \eqref{MAP formulation} lies in the computation of the likelihood function therein which requires marginalization over the latent variables: $p(\by| \bx) = \int p(\by,\bz|{\bx}) {\text d}\bz$. Essentially, the marginalization leads to undesired integrals inside the logarithm operation before the likelihood function, making the direct optimization on \eqref{ML formulation} or \eqref{MAP formulation} intractable. Furthermore, one can observe that given a large set of latent variables (corresponding to a high-dimension space for $\bz$), such marginalization involving many nested integrals can be computationally demanding. We tackle the challenge by applying theories of EM and variational inference in the algorithmic design.

\section{BackSense Design: Randomized Sensor Transmission}\label{Section:BackscatterSensingDesign}
In this section, the design of random-encoding scheme is presented that solves the design problem formulated in Section \ref{randomized transmission problem}. Its implementation on the BackSense architecture is illustrated in Fig.~\ref{Fig:Architecture}.  Then the distributions of the randomized sensor signals are analyzed. The results facilitate  the design of reader inference algorithms in the following two sections. 

\begin{figure}[t]
\centering
\includegraphics[width=17cm]{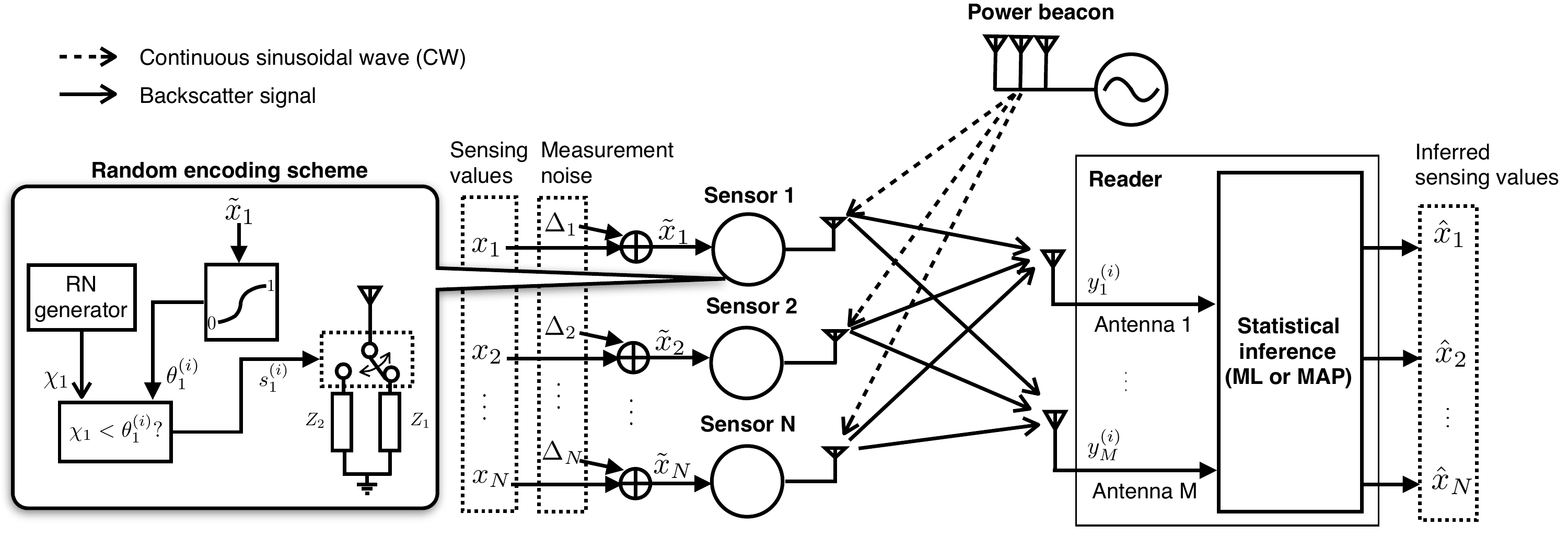}
\caption{The design of  BackSense system.}
\vspace{-6mm}
\label{Fig:Architecture}
\end{figure}

\subsection{random-encoding scheme}\label{Architecture}
As shown in Fig. \ref{Fig:Architecture}, the random-encoding scheme consists of the following two key steps. 

\vskip 10pt
\noindent
\textbf{Step 1) Range normalization}. Distributed as a continuous variable in real domain, a measured sensing value has an infinite range. The goal of the current step is to map it to a new variable with a unit range so as to facilitate random encoding in Step 2. Consider the sensor measurement  $\tilde{x}$ of a sensing value $x$. Let $\theta$ be a variable within the range $[0, 1]$. Then $\tilde{x}$ can be one-to-one mapped to $\theta$ using the well-known  sigmoid function (a function characterized by a S-shaped curve), denoted as $\mathcal{F}$: 
\begin{align}\label{mapping}
\theta=\mathcal{F}(\tilde{x})=\frac{1}{1+\exp\l(-\frac{\tilde{x}-\mu}{\sigma}\r)}, 
\end{align}
where $\mu$ and $\sigma$ are mean and variance of  $x$, respectively. This gives the name of $\theta$ as the \emph{mapped sensing value}. It is  assumed that the parameters $\mu$ and $\sigma$ are known by the sensor and reader via estimation using historical sensing data. Consider the sensing values and their measurements  for all $N$ sensors, represented by the vectors  $\mathbf{x}$ and $\mathbf{\tilde{x}}$, respectively, with $\mathbf{\tilde{x}}=\mathbf{x}+\boldsymbol\Delta$. The vector $\boldsymbol\Delta$ groups i.i.d. measurement-noise r.v. The element-wise mapping using \eqref{mapping} is represented by $\mathcal{F}(\mathbf{\tilde{x}})$,  where the mapping of the $n$-th element is specified by the corresponding parameters $\mu_n$ and $\sigma_n$. The result is denoted as  $\boldsymbol{\theta}$, namely  $\boldsymbol{\theta} = \mathcal{F}(\mathbf{\tilde{x}})$. Since the mapping $\mathcal{F}$  is one-to-one, the mapped sensing values $\boldsymbol{\theta}$ are used in place of $\mathbf{\tilde{x}}$ in the subsequent analysis and design.  

\noindent
\textbf{Step 2) Encoding by comparison with a uniform distributed r.v.} Consider an arbitrary sensor. As shown in Fig. \ref{Fig:Architecture}, the mapped sensing value $\theta$ is compared with a r.v., denoted as $\chi$,  uniformly distributed in the same range of $[0, 1]$. The r.v. is generated independently for different symbol durations.  Then, the comparison yields a random sequence of binary bits for turning on/off backscattering of the sensor (see Fig. \ref{Fig:Architecture}). The r.v.  $\chi$ can be generated locally using a simple circuit (see e.g., \cite{Che2008, Balachandran2008}). Let $S$ denote an arbitrary  symbol transmitted by the sensor using on/off keying. Then $S = \bar{\rho}$ if $\chi  < \theta$ or otherwise $S = 0$. Thereby, the distribution of $S$ is parameterized by the mapped sensing value $\theta$ as follows: 
\begin{equation}\label{Eq:SymDist}
S = \l\{\begin{aligned} 
&\bar{\rho}, && \text{w.p}\  \theta \\
&0, && \text{w.p}\ 1 - \theta. 
\end{aligned}
\r.
\end{equation}

\begin{remark}[How to decode the sensing value?] \emph{Given random encoding, the sensing value can be decoded at the reader by inferring the statistics of many observed symbols transmitted by sensors as illustrated by the following toy example case where all $N$ sensors with identical and fixed channel gains: $h_n^{(i)} = h$ for all $1\leq n\leq N$, $1\leq i \leq L$ and have uniform mapped sensing value $\theta$ to deliver. As a result, the reader receives a set of $L$ symbols transmitted by sensors, $\{\sum_{n=1}^N h_n^{(i)}S_n^{(i)} + w^{(i)}\}$, where $\{w^{(i)}\}$ represent the set of i.i.d. samples of the channel-noise process. If $L$ is sufficiently large and applying the law of large numbers,   then the mapped sensing value can be recovered as 
\[
\theta =\lim_{L\rightarrow \infty } \frac{1}{\bar{\rho}h N L}\sum_{i=1}^L \sum_{n=1}^N h_n^{(i)}S_n^{(i)} + w^{(i)}. 
\]
The measured sensing value is then obtained as $\tilde{x} = \mathcal{F}^{-1}(\theta)$. 
%Next, consider the case with $N$  sensors with identical channel gains: $h_n = h$ for all $1\leq n\leq N$. Given many sensors, the reader can estimate the mapped sensing value $\theta$ from the received signal $\sum_{n=1}^N h_n S_n$ as: 
%\[
%\theta =\lim_{N\rightarrow \infty }\l [ 1- \frac{1}{\bar{\rho}h N}\sum_{n=1}^N (h_nS_n + W_n)\r ]. 
%\]
The above illustration shows how the sensing value can be decoded in a simplified scenario and the performance of sensing-data uploading improves with the growth of observation duration. However, the target scenario is more complex involving practical factors including  transmissions over multi-antenna channels, unknown and fluctuating channel states, spatial variation of sensing values, and finite numbers of observed symbols. This is the reason why the powerful statistical inference is needed for recovering sensing values (see Sections \ref{Section:Withoutnoise} and \ref{Section:Withnoise}). 
}
\end{remark}
 
 \subsection{Distributions of Randomized Transmissions}\label{Section:RandDist}
To facilitate statistical inference  at  the reader,  we derive  in this subsection the distributions of some r.v. arising from randomized transmissions using the proposed random-encoding scheme.

\subsubsection{Distribution of mapped  sensing values} 
First, consider the case without measurement noise. For the current case, the mapping between $\boldsymbol \theta$ and $\bx$ is one-to-one and thus  $\boldsymbol \theta$ can be considered as an equivalence of  $\bx$. Using  the mapping  in \eqref{mapping}, the prior distribution of  $\boldsymbol \theta$, denoted as $p(\boldsymbol \theta)$,  can be derived based on that of $\bx$, i.e., $p(\bx)$ given in \eqref{PDF of x}, as presented below.
\begin{lemma}\label{prop:1}
For the case of heterogeneous sensing values without measurement noise, the distribution of $\boldsymbol \theta$ is specified by
\begin{align}\label{Prior distribution of bf theta}
p(\boldsymbol \theta) = \l(\prod_{n=1}^N \frac{\sigma_n}{\theta_n-\theta_n^2}\r)
 \frac{1}{(2\pi)^{\frac{N}{2}} |\boldsymbol \Sigma_\bx|^{\frac{1}{2}}}\exp\l( -\frac{1}{2} \br(\boldsymbol \theta)^T \boldsymbol \Sigma_\bx^{-1} \br(\boldsymbol \theta)\r),
\end{align}
where $\br(\boldsymbol \theta)$ is an element-wise function of $\boldsymbol \theta$ given by 
\begin{align}
\br(\boldsymbol \theta) = \l[-\sigma_1 \log \l(\frac{1}{\theta_1} - 1\r), \cdots,  -\sigma_N \log \l(\frac{1}{\theta_N} - 1\r)\r]^T. 
\end{align}
\begin{proof}
See Appendix \ref{appendix:prop:1}.
\end{proof}
\end{lemma}
For the simplified case of uniform sensing values (all elements of $\bx$ are equal), the expression given in \eqref{PDF of x} no longer holds as the covariance matrix $\boldsymbol \Sigma_\bx$ becomes singular. Alternatively, note that $\bx$ is specified by the scalar uniform value $x$ only and whose distribution can be modelled by a univariate Gaussian $x \sim {\cal N}(\mu_0,\sigma_0^2)$, where $\mu_0$ and $\sigma_0^2$ denote the mean and variance respectively. Thus the PDF of $x$ is then given by
$p(x) = \frac{1}{\sqrt{2\pi \sigma_0^2}}\exp\l( -\frac{(x - \mu_0)^2}{2\sigma_0^2}\r).$

Accordingly, following the similar procedure as for deriving Lemma \ref{prop:1}, the distribution of $\theta = {1}/\({1 + \exp\l(-\frac{x - \mu_0}{\sigma_0}\r)}\)$ can be derived as follows: 
\begin{align}\label{prior for theta}
p(\theta) = \frac{1}{(\theta-\theta^2)\sqrt{2\pi}}
 \exp\l( -\frac{\l(\log \l( \frac{1}{\theta} -1 \r)\r)^2}{2} \r). 
\end{align}
 
Next, consider the case with measurement noise. For this case, the mapped sensing-value vector  
$\boldsymbol{\theta}$ is a function of both the sensing values $\bx$ and measurement noise 
$\boldsymbol{\Delta}$. The randomized effect of the measurement noise on the mapping between  $\boldsymbol{\theta}$ and $\bx$ is captured by the conditional distribution of $\boldsymbol{\theta}$ (a latent variable) given $\bx$ (a parameter) which is required in the subsequent reader inference. The result is derived using a similar method as Lemma~\ref{prop:1} and given as follows. 

\begin{lemma}\label{prop:8}
For the case of heterogeneous sensing values with measurement noise, the conditional distribution function  of $\boldsymbol\theta$ given $\bx$ is given as:
\begin{align}\label{Conditional density of bf theta}
p(\boldsymbol \theta | \bx) = \l(\prod_{n=1}^N \frac{\sigma_n}{\theta_n-\theta_n^2}\r)
 \frac{1}{(2\pi)^{\frac{N}{2}} \delta^N}\exp\l( -\frac{\sum_{n=1}^N (-\sigma_n \log \l(\frac{1}{\theta_n} -1\r) + \mu_n - x_n)^2}{2\delta^2} \r).
\end{align}
\end{lemma}

\subsubsection{Distribution of the number of active sensors} A latent variable that affects the distribution of received signal at the reader is the number of active (on) sensors, denoted as $T^{(i)}$ for the $i$-th slot (see Fig. \ref{Fig:Probabilistic model without noise} and \ref{Fig:Probabilistic model with noise}). Given on/off keying, $T^{(i)}$ is the sum over a sequence of Bernoulli r.v. that represent the $i$-th symbols transmitted by all sensors: $T^{(i)} = \sum_{n=1}^N S^{(i)}_n$ where $\{S^{(i)}_n\}$ follow the distribution in  \eqref{Eq:SymDist}. Consider the case where the mapped sensing values differ over sensors, denoted as $\theta_1, \theta_2, \cdots, \theta_N$. The r.v. is a Poisson-Binomial r.v. that describes the numbers of successes in $N$ independent Bernoulli trails with different individual success probabilities, denoted by ${\cal PB}(N, \theta_1, \theta_2, \cdots, \theta_N)$. The corresponding \emph{probability mass function} (PMF) is  given as \cite{fernandez2010closed}:
\begin{align}\label{Poisson Binomial}
p(T^{(i)} = m | \boldsymbol \theta) = \frac{1}{N+1}\sum_{\ell=0}^N c^{-\ell m} \prod_{n=1}^N[1+(c^\ell - 1)\theta_n], \qquad  m = 0,1,\cdots, N,
\end{align}
where $c = \exp \left( \frac{2\pi j}{N+1} \right)$ with $j = \sqrt {-1}$. For the special case of uniform mapped sensing values denoted as  $\theta$, the distribution of $T^{(i)}$ reduces to the  Binomial distribution: 
\begin{align}\label{Binomial}
p(T^{(i)} = m | \theta) = {N \choose m}\theta^{m}(1-\theta)^{N-m}, \qquad  m = 0,1,\cdots, N.
\end{align}

\section{BackSense Design: Reader Inference with Ideal Measurements}\label{Section:Withoutnoise}
This section addresses the design problem of reader inference formulated in Section \ref{inference problem}. To this end, statistical inference algorithms are designed based on the EM framework, which enable the reader to recover the sensing values form the received signals. Ideal measurements are assumed and the effect of measurement noise on the design is addressed in the next section. We start with a brief introduction to the basic principle of the EM framework. Then a Bayesian network is constructed for specifying  the statistical dependencies between the observations and unknowns (including the latent variables and model parameters). Based on the Bayesian network, EM algorithms targeting BackSense are designed for both the criteria of ML and MAP. 

\subsection{Principle of the EM Framework}\label{EM framework}
Due to the existence of latent variables, directly solving the potentially non-convex ML and  MAP problems formulated in Section \ref{inference problem} is intractable. Alternatively, the EM framework tempts to find local-optimal  solutions by an iterative procedure, involving iterations between two main steps, i.e., the \emph{Expectation step} (E-step) and the \emph{Maximization step} (M-step) (see e.g., \cite{gupta2011theory}). In the E-step, the expectation of the \emph{complete-data log likelihood} is used  as a surrogate of the required \emph{incomplete-data log likelihood} as defined later. Then in the  M-step, the surrogate is  maximized instead of the actual incomplete-data log likelihood. The use of the surrogate makes finding the solution more tractable without compromizing its local-optimality. For exposition, the  details are discussed in the sequel based on the ML criterion while the same principle also applies  to the MAP criterion. 

Assume that we have  $L$ i.i.d. observations. We denote all of the observed variables as $\bY = [\by^{(1)},\by^{(2)},\cdots,\by^{(L)}]$,  all of the latent variables as $\bZ = [\bz^{(1)},\bz^{(2)},\cdots,\bz^{(L)}]$, and all of the parameters as $\bf \Phi$.  By introducing an auxiliary distribution $q(\bZ)$ defined over the latent variables $\bZ$, the \emph{incomplete-data log likelihood} can be decomposed by
\begin{align}\label{decomposition}
\log p(\bY|{\bf \Phi}) = {\cal L}(q,{\bf \Phi}) + \text{KL}(q||p),
\end{align}
where ${\cal L}(q,{\bf \Phi}) = {\int} q(\bZ) \log \left\{ \frac{p(\bY,\bZ| {\bf \Phi})}{q(\bZ)}\right\} {\text d}{\bZ}$ is the expectation of the complete-data log likelihood and the $\text{KL}(q||p) = - \int q(\bZ) \log \left\{ \frac{p(\bZ|\bY,{\bf \Phi})}{q(\bZ)} \right\} {\text d}{\bZ}$ is known as \emph{Kullback-Leibler} (KL) divergence. Note that the KL divergence measures  the ``similarity'' between the two  distributions of $q$ and $p$. Hence the measure    is a non-negative function attaining  the  minimum when $q$ and $p$ are identically distributed. Therefore, it follows from \eqref{decomposition} that ${\cal L}(q,{\bf \Phi})$ lower bounds  $\log p(\bY|{\bf \Phi})$. One can see that the bound reaches equality, namely $\text{KL}(q||p) = 0$, by letting $q(\bZ) = p(\bZ|\bY,{\bf \Phi})$. This   essentially gives the E-step. The resultant tight lower bound can then serve as a surrogate of $\log p(\bY|{\bf \Phi})$ that is  optimized to provide updated parameters $\bf \Phi$, constituting the M-step. As a result, the maximization of \eqref{decomposition} can be solved by repeating the following EM cycle, where the subscript $(t)$ is used to denote the iteration index:   
\begin{align}
{\bf E\!-\!step}:\qquad 
&q^{(t+1)}(\bZ) \gets \arg \max_q {\cal L}(q,{\bf \Phi}^{(t)})\label{Estep update}\\
\implies &q^{(t+1)}(\bZ) = p(\bZ|\bY,{\bf \Phi}^{(t)}). \label{Estep update2}\\
{\bf M\!-\!step}:\qquad  &{\bf \Phi}^{(t+1)} \gets \arg \max_{\bf \Phi} {\cal L}(q^{(t+1)},{\bf \Phi})\label{Mstep update}\\
\implies &{\bf \Phi}^{(t+1)} \gets \arg \max_{\bf \Phi} \int p(\bZ|\bY,{\bf \Phi}^{(t)})\log p(\bY,\bZ|{\bf \Phi}^{(t)}) {\text d}\bZ. \label{Mstep update2}
\end{align}
The above EM cycle starts from some initial values for the parameters ${\bf \Phi}$ and keep repeating until convergence. The convergence of the EM framework to at least a local optimum is guaranteed since both the E-step \eqref{Estep update} and M-step \eqref{Mstep update} lead to a non-decreasing objective ${\cal L}(q,{\bf \Phi})$ at each EM cycle \cite{gupta2011theory}.

\subsection{Construction of Bayesian Network}

Designing the BackSense inference algorithms begins with constructing a Bayesian network that specifies the set of parameters and latent variables affecting the received signal as well as their interdependence characterized by their conditional distributions. The set of parameters are unknown to the readers and thus the target of inference. The parameter set  $\bf \Phi$ comprises the mapped sensing values $\boldsymbol{\theta}$, channel variance $\sigma_h^2$ and channel-noise variance $\sigma_w^2$: ${\bf \Phi} = \{{\boldsymbol \theta}, \sigma_h^2, \sigma_w^2\}$. For the current case without measurement noise, there exists only one type of latent variables, namely the set of randomized symbols transmitted by sensors, that affect the observed variable, namely the received symbols over multiple antennas. 

Consider the $i$-th slot. To complete the construction of the Bayesian network, the conditional distribution of the received symbol vector $\by^{(i)}$ given the transmitted symbol vector $\bs^{(i)}$ is derived as follows. Note that the elements of the additive noise and the channel matrix in \eqref{Eq:Signalmodel} are complex Gaussians. It follows from \eqref{Eq:Signalmodel} that  given $\bs^{(i)}$, $\by^{(i)}$ is a sum of complex Gaussian vectors that also follows the complex Gaussian distribution. Specifically, conditioned on $\bs^{(i)}$, $\by^{(i)}  \sim {\cal CN}(0,\sigma_h^2 \sum_{n=1}^N S^{(i)}_n + \sigma_w^2)$. One can see that given on/off keying for sensor transmissions, the sum $\sigma_h^2 \sum_{n=1}^N S^{(i)}_n$ is equal to $\bar{\rho} T^{(i)}$ where $T^{(i)}$ is the number of active tags and $\bar{\rho}$ is the reflection coefficient for an active sensor. Without loss of generality, we set $\bar \rho = 1$ in the subsequent analysis for ease of notation. Mathematically, 
\begin{align}\label{Eq:ConditionalPDF}
p(\mathbf{y}^{(i)}\mid \bs^{(i)}) = p(\mathbf{y}^{(i)}\mid T^{(i)})
=\frac{1}{\left [2\pi (\sigma_h^2T^{(i)} + \sigma_w^2)\right]^M} \exp\l(-\frac{{\| \mathbf{y}^{(i)}\|_2^2}}{2(\sigma_h^2T^{(i)}+\sigma_w^2)}\r).
\end{align}
This suggests that the latent variable $\bs^{(i)}$ can be replaced with $T^{(i)}$ without changing the distribution of $\by^{(i)}$. The distribution of $T^{(i)}$ is derived earlier as given in \eqref{Poisson Binomial} and \eqref{Binomial}. Note that the variable replacement reduces the dimensions of latent space from $N$ to just $1$, reducing the complexity of inference.   Based on the above discussion, the Bayesian network corresponding to the $i$-th received symbol is illustrated in Fig.~\ref{Fig:Probabilistic model without noise}. 

\begin{figure}[tt]
\centering
\includegraphics[width=10cm]{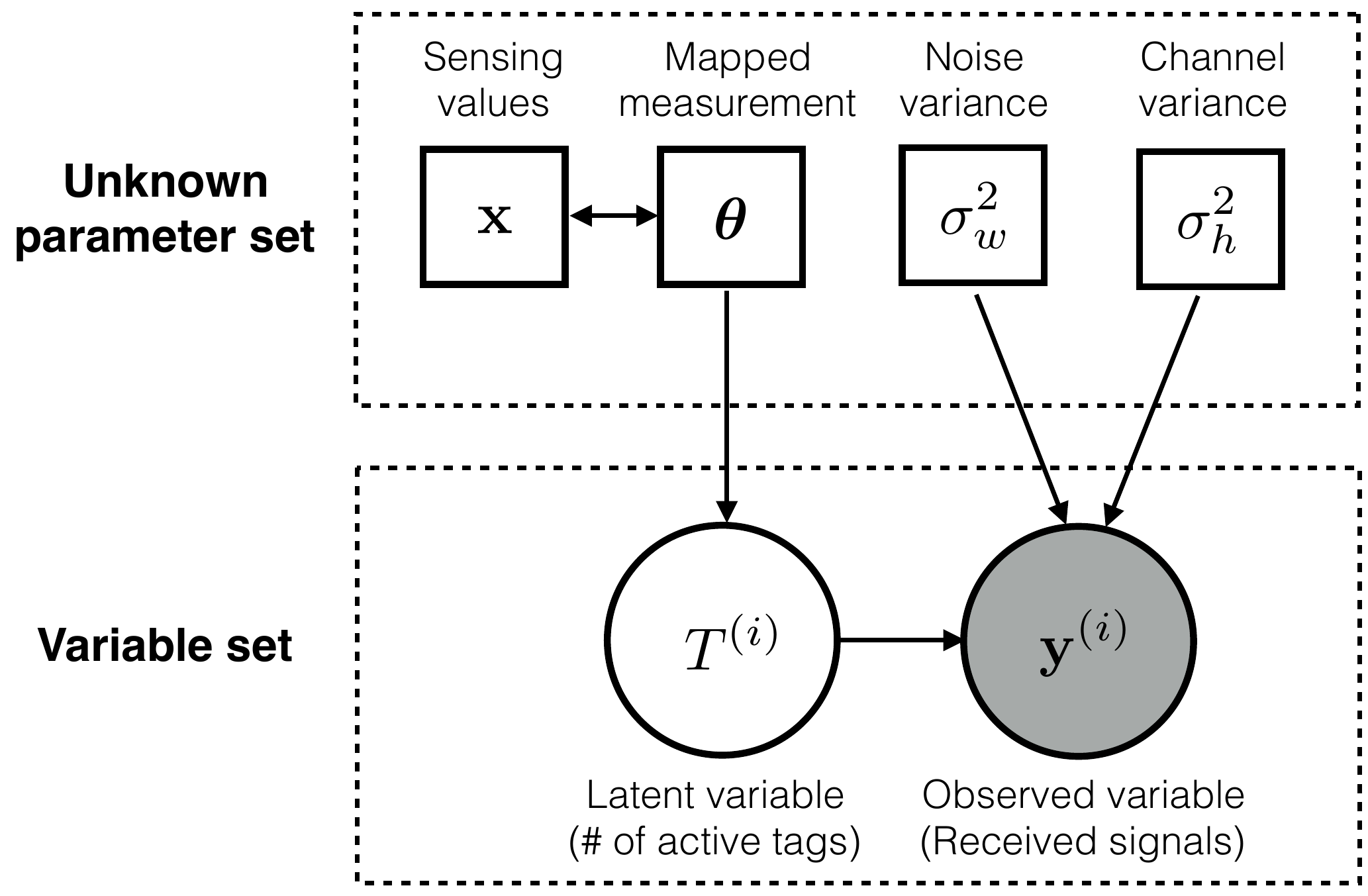}
\caption{Bayesian network for BackSense inference without measurement noise.}
\label{Fig:Probabilistic model without noise}
\vspace{-6mm}
\end{figure}

\subsection{Reader ML Inference with Uniform Sensing Values}
 In this subsection, we consider the simple case of uniform sensing values $\theta$ and the ML criterion for designing the reader-inference algorithm. The design is subsequently extended to the case of heterogeneous  sensing values as well as the MAP criterion. Based on the EM framework introduced in Section \ref{EM framework}, the algorithm for inferring $\theta$   is designed for BackSense by driving the corresponding E-step and M-step shortly. The resultant procedure is summarized in Algorithm \ref{algorithm:1}. Given the inferred sensing value $\hat{\theta}$ and the mapping in \eqref{mapping}, the inferred sensing value is then obtained as
%$\hat x = \sigma_0 \log \(\frac{1}{\hat \theta} - 1\) + \mu_0.$
\begin{align}\label{inverse mapping}
\hat x = \sigma_0 \log \(\frac{1}{\hat \theta} - 1\) + \mu_0.
\end{align}

\vspace{5pt}
\noindent\underline{a) Derivation of  E-step} \newline
Consider the Bayesian network in Fig. \ref{Fig:Probabilistic model without noise}. Given the  parameter set  ${\bf \Phi}_{\sf uni} = \{\theta, \sigma_h^2, \sigma_w^2\}$, the required posterior distribution of the latent variable $T^{(i)}$ can be computed via Bayes' Law as
\begin{align}
q(T^{(i)} = m) = p(T^{(i)} = m \mid \by^{(i)},{\bf \Phi}_{\sf uni}) = \frac{f_m}{\sum_n f_n},\label{Estep update3}
\end{align}
where $f_m$ is defined as 
\begin{align}\label{definition of fm}
f_m = \frac{1}{\left [2\pi (\sigma_h^2 m + \sigma_w^2)\right]^M} \exp\l(-\frac{{\| {\by^{(i)}}\|_2^2}}{2(\sigma_h^2 m+\sigma_w^2)}\r) {N \choose m}\theta^{m}(1-\theta)^{N-m}, \;\; m = 0,1,\cdots,N. 
\end{align} 
This completes the derivation of the E-step as indicated in \eqref{Estep update2}. For ease of notation, let $q^{(i)}_T(m)$ denote $q(T^{(i)} = m)$ in the rest of the paper. 

%-------------------------
 \begin{algorithm}[t!]
\textbf{Initialization}:

Initialize the model parameters $\theta$, $\sigma_h^2$, $\sigma_w^2$ properly.

\textbf{Iteration}:

1) \textbf{E-step}: Update the posterior distribution of the latent variable,  $\{q^{(i)}_T(m)\}$,  according to \eqref{Estep update3}.

2) \textbf{M-step}: Update the parameter $\theta$ and $\sigma_h^2$, $\sigma_w^2$ by substituting the latest values of $\{q^{(i)}_T(m)\}$ into  \eqref{update for theta} and \eqref{update for sigma_h and sigma_w}, respectively.

\textbf{Until Convergence}.
\caption{EM algorithm for ML reader inference}\label{algorithm:1}
\end{algorithm}
%-------------------------

\vspace{5pt}
\noindent\underline{b) Derivation of M-step} \newline
According to the M-step in  \eqref{Mstep update2}, the parameter set  ${\bf \Phi}_{\sf uni}$ should be optimized by solving 
\begin{align}\label{Mstep uniform ML case}
 {\bf \Phi}_{\sf uni}^* = \arg \max_{{\bf \Phi}_{\sf uni}} \sum_{i=1}^L \sum_{m=0}^N q^{(i)}_T(m) \log p(\by^{(i)},T^{(i)} = m \mid {{\bf \Phi}_{\sf uni}}),
\end{align}
where the joint distribution over the observed and latent variables is given by
\begin{align}\label{Joint distribution}
p(\by^{(i)},T^{(i)} = m \mid {{\bf \Phi}_{\sf uni}}) = f_m, \qquad m = 0,1,\cdots, N,
\end{align}
where $f_m$ is defined in \eqref{definition of fm}. Recall that the parameter set ${\bf \Phi}_{\sf uni} = \{\theta, \sigma_h^2, \sigma_w^2\}$. A close observation of  \eqref{Mstep uniform ML case} and \eqref{Joint distribution} reveals that $\theta$ and $\{\sigma_w^2, \sigma_h^2\}$ are decoupled in the objective function in \eqref{Mstep uniform ML case} and thereby they can be optimized  separately. 

Firstly, we derive the updating formula for  the parameter $\theta$ by considering those terms in \eqref{Mstep uniform ML case} that are related to $\theta$ only, the corresponding optimization problem is given as 
\begin{align}\label{Optimization for updating theta}
\theta^* = \arg \max_{\theta \in (0,1)} \sum_{i=1}^L \sum_{m=0}^N q^{(i)}_T(m) [m \log \theta + (N - m) \log (1-\theta)],
\end{align}

\begin{subtheorem}{thm}
\begin{thm}\label{prop:3}
The value of  $\theta$ that solves the optimization problem in \eqref{Optimization for updating theta} is given by
\begin{align}\label{update for theta}
\theta^* = \frac{1}{LN}{\sum_{i = 1}^L \sum_{m = 1}^N q^{(i)}_T(m) m}. 
\end{align}
\end{thm}
\begin{proof}
See Appendix \ref{appendix:prop:3}.
\end{proof}

The above result yields  the formula for updating the  parameter $\theta$ in the M-step.

Next, to derive the updating formula for the parameters $\sigma_w^2$ and $\sigma_h^2$ , group  the relevant   terms in the objective function in  \eqref{Mstep uniform ML case}.  The corresponding optimization problem is 
\begin{align}\label{Optimization for updating sigma}
\max_{\sigma_h^2, \sigma_w^2} \sum_{i=1}^L \sum_{m=0}^N q^{(i)}_T(m) \l[-M \log (\sigma_h^2 m+\sigma_w^2) -  \frac{\| {\by^{(i)}}\|_2^2}{2(\sigma_h^2 m+\sigma_w^2)} \r].
\end{align}
The problem is non-convex and thus it is difficult to find its solution in closed form. To overcome the difficulty, we develop a sub-optimal approach that relaxes the problem to a quasiconvex one and solving it provides closed-form updates for the target parameters $\sigma_h^2$ and $\sigma_w^2$. For this purpose, one can observe that the objective function depends on  $\sigma_h^2$ and $\sigma_w^2$ via a weighted sum, defined  as  $\Sigma_m = \sigma_h^2 m+\sigma_w^2$. By substituting $\Sigma_m$,  the optimization problem in \eqref{Optimization for updating sigma} is converted to the following quasiconvex problem: 
\begin{align}\label{Optimization for updating Sigma_t}
\max_{\{\Sigma_m\}} \sum_{i=1}^L \sum_{m=0}^N q^{(i)}_T(m) \l(-M \log (\Sigma_m) -  \frac{\| {\by^{(i)}}\|_2^2}{2\Sigma_m} \r).
\end{align}
\begin{thm}\label{prop:4}
The optimal  values of $\{\Sigma_m\}$ that solve the problem in \eqref{Optimization for updating Sigma_t} are given by
\begin{align}\label{update for Sigma_t}
\Sigma_m^* = \frac{\sum_{i = 1}^L  q^{(i)}_T(m) \| {\by^{(i)}}\|_2^2}{2M\sum_{i = 1}^L  q^{(i)}_T(m)}, \qquad m = 0,1,\cdots, N. 
\end{align}
\end{thm}
\begin{proof}
See Appendix \ref{appendix:prop:4}.
\end{proof}
Note that in general, it is infeasible to find the values of $\sigma_h^2$ and $\sigma_w^2$ that satisfy $(N+1)$ equations: $\Sigma_m = \sigma_h^2 m+\sigma_w^2$ with $m = 0, 1, \cdots, N$. Thus, we resort to finding approximate values, denoted as $(\sigma_w^{2*}, \sigma_h^{2*})$,  by \emph{linear regression}: 
\begin{align}
(\sigma_h^{2*}, \sigma_w^{2*}) &= \arg \min_{(\sigma_w^2, \sigma_h^2)} \sum_{m=0}^N \l(\sigma_w^2 + m \sigma_h^2 - \Sigma_m^*\r)^2\nn\\
&= \l(\frac{\sum_{m=0}^N (m-\frac{N}{2}) (\Sigma_m^* - \bar \Sigma)}{\sum_{m=0}^N (m-\frac{N}{2})^2},  \bar\Sigma - \frac{N}{2}\sigma_h^{2*}  \r) \label{update for sigma_h and sigma_w} 
\end{align}
where $\bar \Sigma$ denotes the mean of the optimal values of $\{\Sigma_m^*\}$ derived in \eqref{update for Sigma_t}.
The above result gives the formula for updating the  parameters $\sigma_h^2$ and $\sigma_w^2$  for the M-step. Combining this  result and that in Proposition~\ref{prop:3} completes the derivation of the M-step.
\end{subtheorem}

\subsection{Extension: Reader MAP Inference}
The EM algorithm based on the ML criterion as shown in Algorithm \ref{algorithm:1} is extended to the MAP criterion as follows. Comparing \eqref{ML formulation} and \eqref{MAP formulation}, one can see that the objective function for the MAP problem differs from the ML counterpart only by the extra term $\log  p(\boldsymbol \Phi)$,  corresponding to the ``log-prior" of the unknown parameters. Thus, deriving the MAP algorithm requires only modifying the M-step of the ML counterpart that updates the estimation of the parameters. Specifically, the M-step for the MAP algorithm can be written as 
\begin{align}\label{Mstep uniform MAP case}
\boldsymbol{\Phi}_{\sf uni}^* =  \arg \max_{{\bf \Phi}_{\sf uni}} \sum_{i=1}^L \sum_{m=0}^N q^{(i)}_T(m) \log p(\by^{(i)},T^{(i)} = m \mid {{\bf \Phi}_{\sf uni}}) + \log p({\bf \Phi}_{\sf uni}).
\end{align}
Recall that $\boldsymbol{\Phi}_{\sf uni} = \{\theta, \sigma_h^2, \sigma_w^2\}$. The prior distribution of $\theta$ is derived as shown in \eqref{prior for theta}. However, for those of $\sigma_h^2$ and $\sigma_w^2$, in practice, such information is difficult to obtain from historical inference results or estimate due to lack of training signals. A common technique to tackle the difficulty is to replace the required prior distributions with some constants, known as \emph{non-informative priors} (see e.g., \cite{murphy2012machine}). To some extent, the technique reduces the performance gain of MAP interference over ML inference. Using the technique and  prior distribution of $\theta$ in \eqref{prior for theta}, the optimization problem for updating $\theta$ in the M-step for the MAP inference is modified from the ML counterpart in \eqref{Optimization for updating theta} as 
\begin{align}\label{MAP Optimization for updating theta}
\max_{\theta \in (0,1)} \sum_{i=1}^L \sum_{m=0}^N q^{(i)}_T(m) [m \log \theta + (N - m) \log (1-\theta)] -\log \theta(1-\theta) - \frac{\l[\log\l(\frac{1}{\theta} - 1\r)\r]^2}{2}.
\end{align}
\begin{thm}\label{prop:5}
The optimal value of $\theta$ that solves the optimization problem in \eqref{MAP Optimization for updating theta} satisfies 
\begin{align}\label{MAP update for theta}
e^{(LN - 2)\theta^*} = a \l(\frac{1}{\theta^*} - 1\r),
\end{align}
where $a$ is a constant and given by $a = \exp\l(\sum_{i=1}^L \sum_{m=0}^N q^{(i)}_T(m) m -1\r)$. 
\end{thm}
The proof is straightforward and omitted here for brevity. Though $\theta^*$ has no closed form, it can be computed by a simple numerical search. 
On the other hand, due to the use of non-informative priors, the procedure for updating the other parameters, i.e.,  $\sigma_h^2$ and $\sigma_w^2$, remains the same as the ML counterpart.

In summary, the iterative algorithm for MAP inference for  the current case can be modified from Algorithm \ref{algorithm:1} by replacing \eqref{update for theta} in Step 2) with solving the equation in \eqref{MAP update for theta}. 

\begin{remark}\emph{(ML versus  MAP)\label{ML vs MAP}
A comparison between  \eqref{Mstep uniform ML case} and \eqref{Mstep uniform MAP case} reveals that, for MAP inference, the reader's  knowledge of parameters' prior distributions  servers as an additional regularization term in the objective function to avoid overfitting to the potential outliers in the data set. This makes  the performance more robust even when there are errorneous data points in the data set or its size is not large enough. } 
\end{remark}

\subsection{Extension: Heterogeneous  Sensing Values}

In the preceding sub-sections, we designed inference algorithms for the simple case of uniform sensing values. They are extended in this sub-section to the general  case of heterogeneous sensing values. The corresponding parameter $\boldsymbol{\theta}$ is a vector comprising $N$ elements instead of being a scalar in the preceding case. Their joint distribution is derived earlier as shown in \eqref{Prior distribution of bf theta}. Given their relation shown in Fig. \ref{Fig:Probabilistic model without noise}, the change in $\boldsymbol{\theta}$ causes the number of active tags, $T^{(i)}$, to follow the  Poisson-Binomial distribution in \eqref{Poisson Binomial}  instead of the Binomial distribution in the case of uniform sensing values. The above changes in parametric distributions require the EM algorithms for reader inference to be modified as follows. 

\vspace{5pt}
\noindent\underline{a) Derivation of E-step} \newline
According to \eqref{Estep update2}, given initial values of the unknown parameters ${\bf \Phi}_{\sf cor} = \{\boldsymbol \theta, \sigma_h^2, \sigma_w^2\}$, the E-step calculates the posterior distribution of the latent variable by Bayes' law as follows. 
\begin{align}
q(T^{(i)} = m) &= p(T^{(i)} = m \mid \by^{(i)},{\bf \Phi}_{\sf cor})= \frac{g_m}{\sum_{n}g_n}, \label{Estep update correlated sensing}
\end{align} 
where $g_m$ with $m=0,1,\cdots, N$ is defined as 
\begin{align}
g_m &= \frac{1}{\left [2\pi (\sigma_h^2 m + \sigma_w^2)\right]^M} \exp\l(-\frac{{\| {\by^{(i)}}\|_2^2}}{2(\sigma_h^2 m+\sigma_w^2)}\r) \frac{1}{N+1}\sum\limits_{\ell=0}^N c^{-\ell m} \prod_{n=1}^N[1+(c^\ell - 1)\theta_n]. 
\end{align}

\vspace{5pt}
\noindent\underline{b) Derivation of M-step} \newline
The formulas for updating the parameters $\sigma_w^2$ and $\sigma_h^2$ remain the same as in \eqref{update for sigma_h and sigma_w}. Only that for updating $\boldsymbol{\theta}$ needs to be modified as follows. Similar to \eqref{Optimization for updating theta}, by considering only those terms related to $\boldsymbol \theta$ in the M-step optimization, the update of $\boldsymbol \theta$ requires to solve the following two optimization problems based on the ML and MAP criteria, respectively: 
\begin{align}
\text{(ML)}\qquad &\max_{\boldsymbol \theta} \tilde g(\boldsymbol{\theta}) \label{Optimization for updating bf theta:ML},\\
\text{(MAP)}\qquad &\max_{\boldsymbol \theta} \l[\tilde g(\boldsymbol{\theta}) - \sum_{n=1}^N \log (\theta_n - \theta_n^2) -\frac{1}{2} \br(\boldsymbol \theta)^T \boldsymbol \Sigma_\bx^{-1} \br(\boldsymbol \theta)\r],\label{Optimization for updating bf theta:MAP}
\end{align}
where $\tilde g(\boldsymbol{\theta})$ is defined as 
\begin{equation}
\tilde g(\boldsymbol{\theta}) = \sum_{i=1}^L \sum_{m=0}^N q^{(i)}_T(m) \log \left\{ \sum_{\ell=0}^N c^{-\ell m} \prod_{n=1}^N[1+(c^\ell - 1)\theta_n] \right\}.
\end{equation}
It can be observed from the above two optimization problems that the mapped sensing values $\{\theta_n\}$ are coupled with each other, making the problem non-convex and the direct solutions intractable. In the area of statistical inference, a common approach for tackling this challenge, known as the \emph{generalized EM} (GEM) framework (see e.g.,  \cite{neal1998view}), is to update the parameters in a tractable way to increase the M-step objective in \eqref{Optimization for updating bf theta:ML}  or \eqref{Optimization for updating bf theta:MAP}  (e.g., using gradient ascent) instead of exactly maximizing it. Using this technique, each complete EM cycle of the GEM algorithm is guaranteed to increase the value of the log likelihood objective until converging to a local optimum \cite{neal1998view}.
Specifically, the iterative equation for updating ${\boldsymbol \theta}$ can be derived as shown below:
\begin{align}\label{ML update for bf theta}
\boldsymbol \theta^{(t+1)} = \boldsymbol \theta^{(t)} + \alpha \nabla\mathcal{L}(\boldsymbol \theta^{(t)}),
\end{align}
where $\alpha$ is the step size and $\nabla\mathcal{L}(\boldsymbol \theta)$ denotes the derivative of the objective function in either   \eqref{Optimization for updating bf theta:ML} or \eqref{Optimization for updating bf theta:MAP}, representing the gradient direction, and whose the $n$-th element is given by: 
\begin{align}
\l[\nabla\mathcal{L}(\boldsymbol \theta)\r]_n(\text{ML}) &= \sum_{i=1}^L \sum_{m=0}^N q^{(i)}_T(m) \frac{ \sum_{\ell=0}^N c^{-\ell m} (c^\ell - 1)\prod_{n'\neq n}^N[1+(c^\ell - 1)\theta_{n'}]}{ \sum_{\ell=0}^N c^{-\ell m} \prod_{n=1}^N[1+(c^\ell - 1)\theta_n]}.\\
\l[\nabla\mathcal{L}(\boldsymbol \theta)\r]_n(\text{MAP})&= \l[\nabla_{\boldsymbol \theta}\tilde{\cal L}(\boldsymbol \theta)\r]_n(\text{ML})- \frac{1-2\theta_n}{\theta_n(1-\theta_n)} - \frac{\sigma_n }{(\theta_n - \theta_n^2)}\l[\boldsymbol \Sigma_\bx^{-1} \br(\boldsymbol \theta)\r]_n.
\end{align}

In summary, the iterative algorithm for reader inference in the current case can be modified from Algorithm \ref{algorithm:1} by simply replacing \eqref{Estep update3} in the E-step with \eqref{Estep update correlated sensing} and \eqref{update for theta} in M-step with \eqref{ML update for bf theta}.

\section{BackSense Design: Reader Inference with Noisy Measurements}\label{Section:Withnoise}
In the preceding section, the mapped sensing value vector  $\boldsymbol{\theta}$ is a parameter in the Bayesian network (see Fig. \ref{Fig:Probabilistic model without noise}) since it is one-to-one mapped from the sensing-value vector $\bx$ and remains fixed throughout the observation duration. In this section, in the presence of measurement noise, $\boldsymbol{\theta}$ is a random vector  conditioned on $\bx$ and varies from symbol-to-symbol, making it a latent variable in the Bayesian network. This causes  the dimensionality of the latent space increasing linearly with the number of sensors.  To overcome the resultant prohibitive complexity of statistical inference for a large number of sensors, we design practical approximate EM algorithms using the method of \emph{variational inference}. The details are presented targeting ML inference.  The extension to the MAP inference is straightforward and omitted for brevity. 

\subsection{Construction of Bayesian Network}
With the addition of measurement noise, the corresponding Bayesian network is modified from the noise-free counterpart (see Fig. \ref{Fig:Probabilistic model without noise}) as shown in Fig.~\ref{Fig:Probabilistic model with noise} and described as  follows. First of all, the mapped sensing vector $\boldsymbol \theta^{(i)}$ is changed from a parameter to a latent variable. Moreover, an extra parameter $\delta^2$, namely the measurement-noise variance, is included that affects the latent variable $\boldsymbol \theta^{(i)}$ together with the sensing-value vector $\bx$. The corresponding conditional distribution of $\boldsymbol \theta^{(i)}$ for given $\delta^2$ and $\bx$ is derived earlier as shown in  \eqref{Conditional density of bf theta}. On the other hand, the conditional distributions of the other latent variable $T^{(i)}$ and the observed variable  $\by^{(i)}$ are identical as their measurement-noise-free counterparts given in  \eqref{Poisson Binomial} and \eqref{Eq:ConditionalPDF}, respectively. 

\begin{figure}[tt]
\centering
\includegraphics[width=11cm]{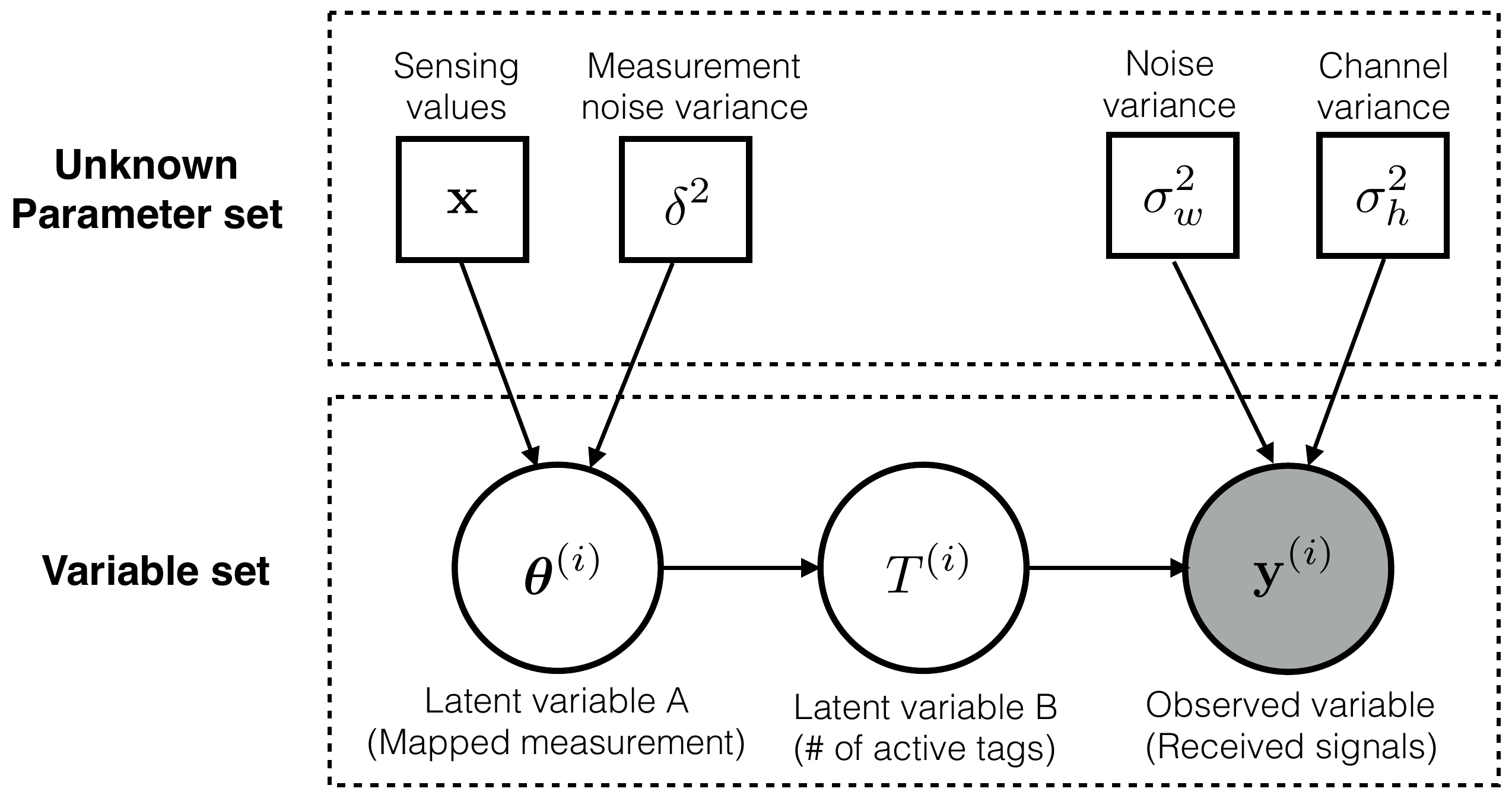}
\caption{Bayesian network for BacksSense with the measurement noise.}
\label{Fig:Probabilistic model with noise}
\vspace{-6mm}
\end{figure}

\subsection{Reader Inference Based on  Approximate EM Implementation}

\subsubsection{Principle of the Variational Inference}\label{Variational Inference} The method of variation inference is one low-complexity approximation of the EM framework. The key idea is to apply the \emph{mean field theory} \cite{kadanoff2009more} that constrains the variational
distribution $q(\bZ)$ in the EM framework [see \eqref{Estep update}: E-step and \eqref{Mstep update}: M-step] to belong to a family of distribution functions charaterized by the factorized form $q(\bZ) = \prod_k q_k(\bz_k)$, where $\bz_k^T$ is the $k$-th row of $\bZ$ and $q_k$ is the corresponding  factor function  \cite{beal2003variational}. When a Bayesian network can be decoupled into sub-networks or the complete data likelihood can be factorized into product form, the approximation can reduce integrals over a high-dimensional  latent space to multiple parallel integrals with low dimensionality and hence low complexity. By substituting the approximate form of $q(\bZ) = \prod_k q_k(\bz_k)$, the two steps of EM framework in \eqref{Estep update2} and \eqref{Mstep update2} can be modified and written  as follows \cite{beal2003variational}: 
\begin{align}
{\bf Variational\;\;E\!-\!step}:\quad 
 &q^{(t+1)}(\bZ)  = b\prod_\ell \exp\left(\int \log p\l(\bY,\bZ|{\bf \Phi}^{(t)}\r)\prod_{k\neq \ell}q_k(\bz_k) d \bz_k\right), \label{VI Estep update}\\
{\bf Variational\;\;M\!-\!step}:\quad   &{\bf \Phi}^{(t+1)} \gets \arg \max_{\bf \Phi} \int q^{(t+1)}(\bZ)\log p(\bY,\bZ|{\bf \Phi}^{(t)}) {\text d}\bZ, \label{VI Mstep update}
\end{align}
where $b$ is a normalization constant whose value can be computed by inspection. One can observe from \eqref{VI Estep update} and \eqref{VI Mstep update} that  the high-dimensional integrals therein can be reduced into low-dimensional ones   if the term $p\l(\bY,\bZ|{\bf \Phi}^{(t)}\r)$ can be decomposed as the product of factors that are functions of non-overlapping subsets of the rows of $\bZ$.  

\subsubsection{Inference Algorithm Derivation} Based on the Bayesian network in Fig. \ref{Fig:Probabilistic model with noise} and applying the mean field approximation, the variational distribution in the EM framework is constrained to take the fully factorized form $q(\bz^{(i)}) = q(T^{(i)}) \prod_n q(\theta_n^{(i)})$. Then according to \eqref{VI Estep update} we have 
\begin{align}
 q^*(T^{(i)} ) &= b_0\exp\l\{ \mathbb{E}_{\boldsymbol \theta} \l[\log p(\by^{(i)}, T^{(i)} , \boldsymbol \theta^{(i)} \mid {\bf \Phi}_{\sf noi}) \r]\r\}, \label{variational distribution T}\\
q^*(\theta_n^{(i)}) &= b_n\exp\l\{ \mathbb{E}_{T,\boldsymbol \theta \backslash \theta_n} \l[\log p(\by^{(i)}, T^{(i)} , \boldsymbol \theta^{(i)} \mid {\bf \Phi}_{\sf noi}) \r]\r\},\qquad \forall n, \label{variational distribution theta}
\end{align}
where $\{b_n\}$ are normalization constants. The explicit expressions  of the approximate marginal distributions of individual latent variables shown in \eqref{variational distribution T} and \eqref{variational distribution theta} can be derived based on the following procedure. Considering \eqref{variational distribution T}, apply the conditional distributions derived earlier as shown in \eqref{Conditional density of bf theta}, \eqref{Poisson Binomial}, \eqref{Eq:ConditionalPDF}  and  together with the chain rule,  followed by separating the terms depending on $T$ in \eqref{variational distribution T} from those  on $\theta_n$. Similar procedure can be applied to \eqref{variational distribution theta}. As the result,  we can obtain the said marginal distributions as shown below, giving the variational E-step. 
\begin{subtheorem}{thm}
\begin{thm} \label{Prop:VE-step}The variational E-step involves calculation of the approximate marginal distributions of latent variables given as: 
\begin{equation}
 q^*(T^{(i)} = m) = \frac{g_a(m)}{\sum_k g_a(k)},\qquad q^*(\theta_n^{(i)} = z_n) = \frac{g_b(z_n)}{\int_0^1g_b(z_n)dz_n}, \;\; \forall n, \label{variational distributions}
\end{equation}
where $g_a(m)$ and $g_b(z_n)$ are given as 
\begin{align}
g_a(m) &= \frac{\mathcal{F}_1(m)}{\sigma^2_w + m \sigma_h^2} \exp \l\{ - \frac{\|\by^{(i)}\|_2^2}{2(\sigma^2_w + m \sigma_h^2)}\r\} ,\nn\\
g_b(z_n)&= \frac{\mathcal{F}_2(z_n)}{z_n (1 - z_n)} \exp\l\{- \frac{1}{2\delta^2}\(-\sigma_n \log\(\frac{1}{z_n} - 1\) + \mu_n - x_n\)^2\r\},  \nn
\end{align}
wherein we define two functions  for ease of notation: 
\begin{align}
\mathcal{F}_1(m) &= \exp\l\{\mathbb{E}_{\boldsymbol \theta} \l(\log \(\sum\nolimits_{\ell=0}^N c^{-\ell m} \prod\nolimits_{k=1}^N[1+(c^\ell - 1)\theta_k]\)  \r)\r\}. \label{Eq:F:1}\\
 \mathcal{F}_2(z_n) &= \left. \exp\l\{\mathbb{E}_{T,\boldsymbol \theta \backslash \theta_n} \l(\log \(\sum\nolimits_{\ell=0}^N c^{-\ell T} \prod\nolimits_{k=1}^N[1+(c^\ell - 1)\theta_k]\)\r)\r\} \right|_{\theta_n = z_n}. \label{Eq:F:2}
\end{align}
\end{thm}

Next, we derive the variational M-step.  To this end, substitute the marginal distributions of latent variables in   \eqref{variational distributions} into \eqref{VI Mstep update}, and solve  the resultant optimization problems.  Thereby, we have  the following set of equations for parametric updating, constituting  the  variational M-step.  

\begin{thm}\label{Prop:VM-step} \emph{For the variational M-step, the parameters $\{\bx, \delta^2, \sigma^{2}_h, \sigma_w^{2}\}$ are updated as:}
 \begin{align}
 x_n^*&=\mu_n - \frac{1}{L}\sum_{i=1}^L \int q(\theta_n^{(i)})
\sigma_n \log\l( {1}/{\theta_n^{(i)}} - 1\r)\text{d} \theta_n^{(i)} \qquad n = 1,2,\cdots, N. \nn \\
 {\delta^2}^* &= \frac{1}{LN} \sum_{i=1}^L   \sum_{n=1}^N \int q(\theta_n^{(i)})(-\sigma_n \log (1/\theta_n^{(i)} -1) + \mu_n - x_n)^2 \text{d} \theta_n^{(i)}. \nn\\
\Sigma_m^* &= \frac{\sum_{i=1}^L q(T^{(i)} =m) \| {\by^{(i)}}\|_2^2}{2M\sum_{i=1}^L q(T^{(i)} =m)}, \qquad m = 0,1,\cdots, N \nn\\
(\sigma^{2*}_h, \sigma_w^{2*}) & = \l(\frac{\sum_{m=0}^N (m-\frac{N}{2}) (\Sigma_m^* - \bar \Sigma)}{\sum_{m=0}^N (m-\frac{N}{2})^2},  \bar\Sigma - \frac{N}{2}\sigma_h^{2*}  \r)  \nn
\end{align}
\end{thm}
\end{subtheorem}
Finally, the main steps of the approximate EM algorithm is summarized in Algorithm \ref{algorithm:2}.
\begin{remark}\emph{(Complexity Reduction by Variational Inference) It can be observed from  Proposition~\ref{Prop:VM-step} that only one-dimension integrals are required  in the variational M-step. However, high dimension integrals are still needed in the variational E-step (see Proposition~\ref{Prop:VE-step}) for the reason that given the parameters, the joint distribution  function of observation and latent variables cannot be decomposed into the fully factorized form. Despite this, the method of variation inference still leads to substantial complexity reduction with respect to the exact EM inference where both the E-step and M-step involve $N$-dimension integrals over the full latent space.} 
\end{remark}

\begin{algorithm}[tt]
\textbf{Initialization}:

Initialize the model parameters $\bx$, $\delta^2$, $\sigma_h^2$, $\sigma_w^2$, and the marginal posteriors $q(T^{(i)})$ and $\{q(\theta_n^{(i)})\}$; 
\textbf{Iteration}:

1) \textbf{E-step}: Compute the approximate marginal posteriors of latent variables  iteratively  by cycling the variational distributions presented in \eqref{variational distributions} in Proposition~\ref{Prop:VE-step} for rounds until convergence; 

2) \textbf{M-step}: Update the parameters according to the formulas in Proposition~\ref{Prop:VM-step}; 

\textbf{Until Convergence}.

\caption{Summary of the approximate EM algorithm using variational inference}\label{algorithm:2}
\end{algorithm}

\section{Simulation Results}\label{Section:StatisticsInference}
For  simulation, the performance metric is  the inference  error, defined by $\frac{\|\hat \bx - \bx\|}{\|\bx\|}$.
The simulation parameters are set as follows unless specified otherwise. The average transmit \emph{signal-and-noise ratio} (SNR), defined as ${P_t}/{\sigma_w^2}$, is set to be $10$ dB.
The variances of the wireless channel and the additive Gaussian noise are $\sigma_h^2 = \sigma_w^2 = 1$.  For the prior distribution of $\bx$, the mean is $\mu_n = 25, \forall n$, and  the covariance matrix $\boldsymbol \Sigma_\bx$ can be decomposed as  $\boldsymbol \Sigma_\bx$ = $\bD^{1/2}\bC_\bx\bD^{1/2}$, where $\bD$ is a diagonal matrix with the diagonal elements collecting the variances of the elements of $\bx$, namely $[\bD]_{n,n} =  \sigma_n^2$, and we set $\sigma_n^2 = 1,\forall n$. Moreover, the correlation matrix $\bC_\bx$ is assumed to  have the \emph{linear exponent autoregressive structure} \cite{simpson2010linear}:
 \begin{align}
\bC_\bx= \left[ {\begin{array}{*{20}{c}}
1&\rho&{{\rho^2}}& \cdots &{{\rho^{N - 1}}}\\
\rho&1&\rho& \cdots &{\rho^{N - 2}}\\
 \vdots & \vdots & \vdots &{}& \vdots \\
\rho^{N-1}&\rho^{N-2}&\rho^{N-3}& \cdots &1
\end{array}} \right], \notag
\end{align}
in which the parameter $\rho \in [0,1)$ controls the level of  correlation between the adjacent sensor tags, the correlation between tags further away  is assumed to decrease with  the growth of their spatial separation following a geometric sequence $1,\rho,\rho^2,\cdots$.

\subsection{Effect of  Observation Duration}
%\begin{figure}[tt]
%\centering
%\includegraphics[width=9cm]{Figures/Effect_of_data_set.pdf}
%\caption{Decoding error versus the size of data set under various parameter settings.}
%\label{Fig:Effect of the size of data set}
%\end{figure}

Consider the scenario without the measurement noise with $N = 4$ and $\rho = 0.5$. Increasing the observation duration, namely  $L$ slots (or symbol durations), generates a larger data set and thereby improves the inference accuracy. The gain is evaluated  in Fig \ref{Fig:Effect of the size of data set}. It is observed that  the inference errors  for both the cases of ML and MAP inference are monotonic decreasing functions of  $L$. MAP inference outperforms  ML especially  in the regime with a relatively short observation duration (small $L$) and the difference  diminishes  as $L$  increases. In this  data-deficient regime, ML tends to  overfit the potential outliners arising from noise corruption. On the other hand,  overfitting is mitigated by MAP due to  a regularization term that attempts to align inference results  with the prior distribution (see Remark \ref{ML vs MAP}). The performance convergence between ML and MAP in the data-sufficient regime (large $L$) is due to the increasingly dominant  contribution from the likelihood function to the posterior objective as indicated in \eqref{MAP formulation}, which reduces that  of the regularization imposed by the prior distribution. Furthermore, the inference accuracy is also observed to increase with the number of reader's receive antenna  $M$. The multi-antenna gain  arises from  the diversity gain resulting from the increased observation dimension.  The same observations also hold for  the scenario with measurement noise.

%\hl{XXX: Modify the figure labels to have "Inference Error".}

\begin{figure*}[tt]
  \centering
  \subfigure[Effect of Observation Duration]{\label{Fig:Effect of the size of data set}\includegraphics[width=0.46\textwidth]{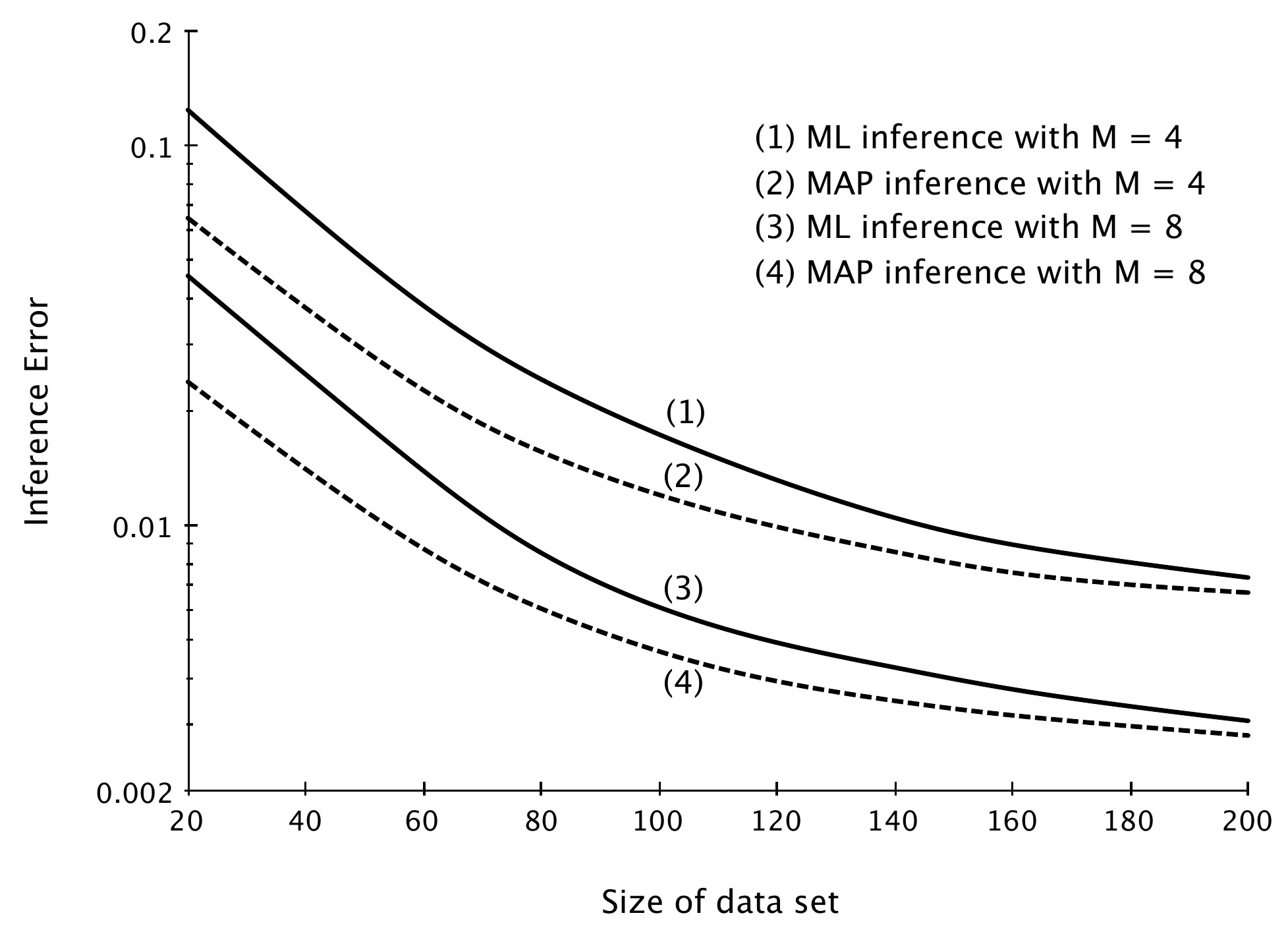}}
  \hspace{0.35in}
  \subfigure[Effect of Spatial Correlation of Sensing Values]{\label{Fig:Effect of the correlation coefficient}\includegraphics[width=0.46\textwidth]{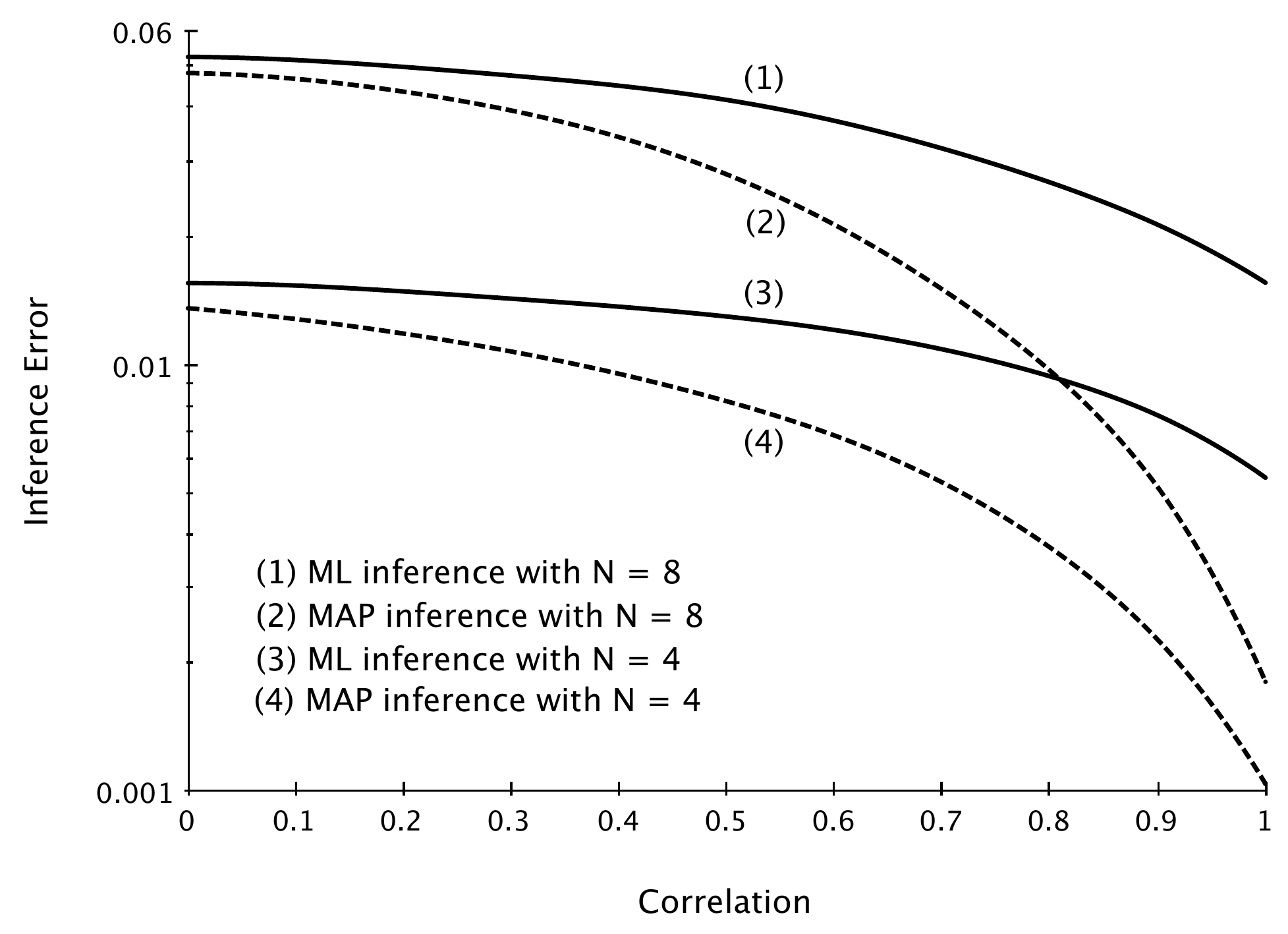}}
    \hspace{0.35in}
  \subfigure[Effect of Average Transmit SNR]{\label{Fig:Effect of SNR}\includegraphics[width=0.46\textwidth]{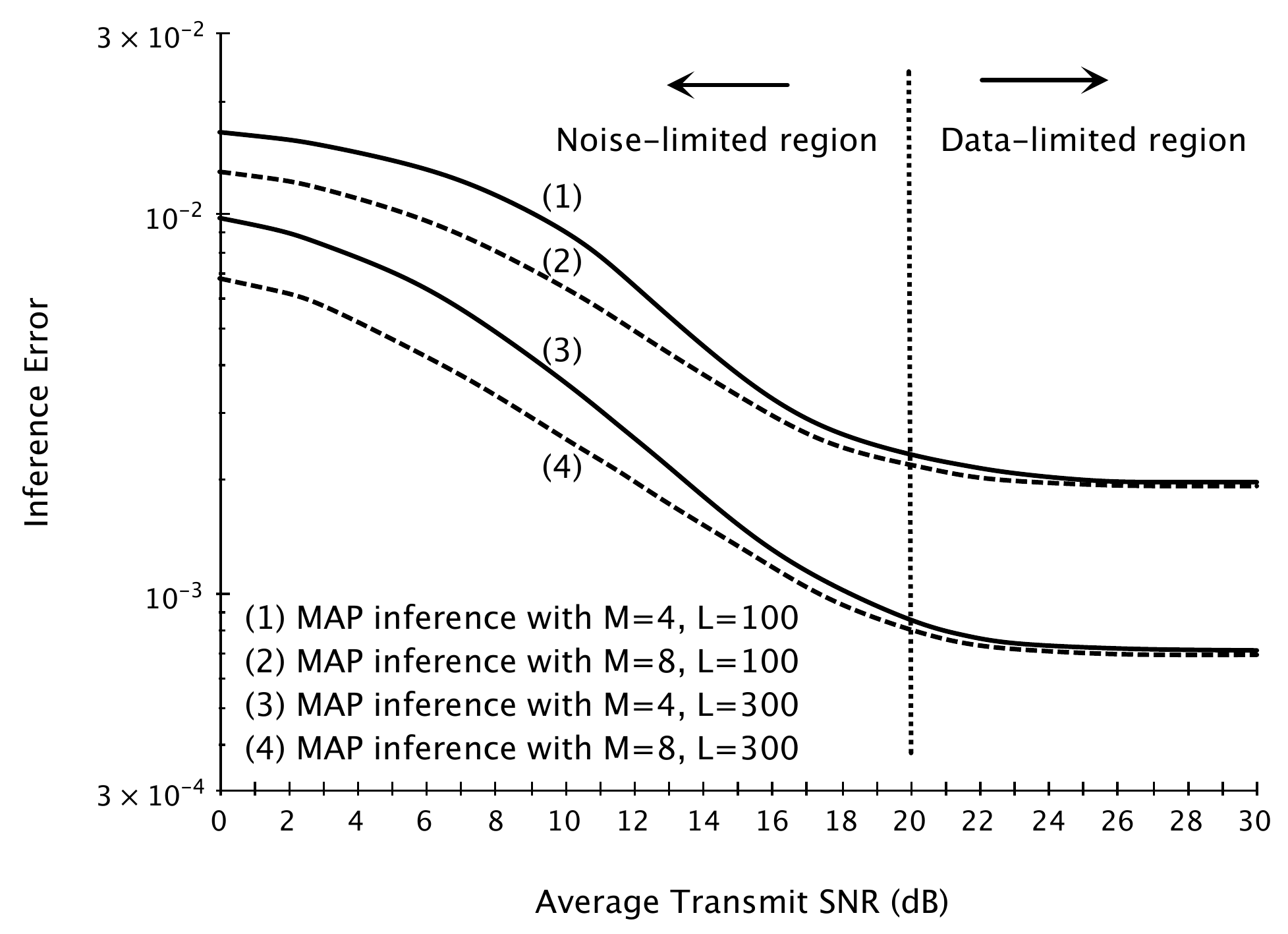}}
    \hspace{0.35in}
  \subfigure[Effect of Measurement Noise]{\label{Fig:Tightness of the VI approximation}\includegraphics[width=0.46\textwidth]{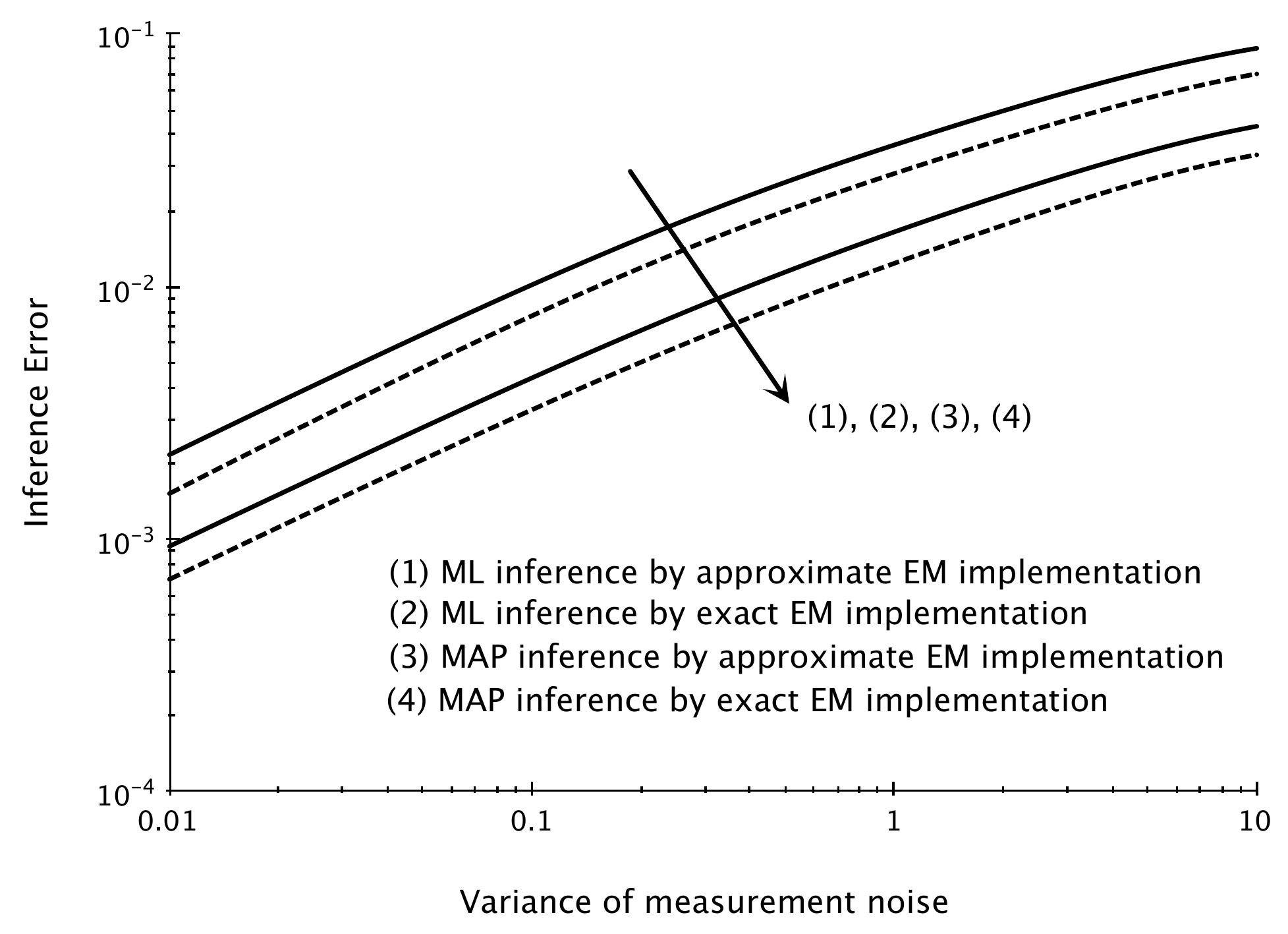}}
  \caption{The effects of system parameters on the accuracy of reader's statistical inference. }
  \label{Fig:Performance evaluation}
  \vspace{-3mm}
\end{figure*}

\subsection{Effect of Spatial Correlation of Sensing Values}
The impact of the correlation coefficient $\rho$ on the inference  accuracy is evaluated in Fig. \ref{Fig:Effect of the correlation coefficient}. The measurement noise is assumed negligible, and we set $M = 4$ and $L = 100$. One can observe from the curves that spatial correlation between sensing values helps reader inference by  improving its accuracy. In particular, MAP inference sees a much more significant performance improvement compared with the ML counterpart as $\rho$ increases. 
The  reason is that the correlation information is explicitly encoded in the prior distribution and exploited by the MAP inference but not ML. In other words, the accessibility of the prior distribution of the environment variables plays a critical role in determining how much gain can be exploited from the spatial correlation between them. The above observation  hold for different $N$ settings. However,  increasing $N$ leads to a decline in decoding accuracy as the number of unknown to be inferred by the reader grows. 

\subsection{Effect of Average Transmit SNR}
Fig. \ref{Fig:Effect of SNR} displays the curves of inference error versus the average transmit SNR for different  values of $M$ and $L$. The figure only shows the curves for MAP inference  without measurement noise. Simulation results for other cases (e.g., ML inference and the presence of measurement noise) are omitted as they all show similar trends. One can observe that as the SNR increases, the inference error sees  fast decrease in  the low-to-moderate SNR regime. However when the SNR is large (exceeds about $20$ dB), the  accuracy cannot be further improved by simply increasing the transmission  power. The reason is that in the high SNR regime the bottleneck of accurate inference is no longer the noise but the data-set size. Increasing the size (or equivalently the observation duration) can lower the error-saturation  level. Furthermore, it is also noted that increasing the number of receive antenna $M$ can considerably alleviate the distortion from noise when SNR is low due to the said diversity gain, but still suffer from the same error floor in the high SNR regime.

\vspace{-3mm}
\subsection{Tightness of Variational-Inference Approximation}
Consider the scenario with  presence of measurement noise, resulting in the dimensionality of the latent space linearly growing with the number of sensors. The performance of the exact EM implementation is compared with its approximate implementation by variational inference in Fig. \ref{Fig:Tightness of the VI approximation} over varying the variance of measurement noise. It is observed that, for both the cases of ML and MAP inference, the approximation leads to only a marginal performance loss compared with   exact EM but leads to  reduced computation complexity. This confirms the  effectiveness of variational inference. Moreover, one can also observe that the inference  performance is satisfactory  even in the regime of strong measurement noise (e.g., the error is less than $0.1$ when $\delta^2 = 10$). This shows  the robustness of reader inference  against the corruption of  measurement noise.

%The tightness of the approximation can be observed throughout the considered range of $\delta$,

%The monotonic increasing error curves with the growth of the measurement noise variance is also observed, aligning with our intuition. 

%\begin{figure}[tt]
%\centering
%\includegraphics[width=9cm]{Figures/Tightness_of_VI.pdf}
%\caption{Comparison between the exact EM implementation and its approximate implementation by variational inference.}
%\label{Fig:Tightness of the VI approximation}
%\end{figure}

%\subsection{BackSensing in the Absence of Measurement Noise}
%
%\subsection{BackSensing in the Presence of Measurement Noise}

\vspace{-3mm}
\section{Concluding Remarks}\label{Section:Conclusion}
The work presents a novel design framework for backscatter sensor systems. The design features the application of machine learning to tackle hardware and signal processing limitations of backscatter sensors as well as reducing overhead (due to e.g., scheduling protocol and channel estimation and feedback). This work represents an initial investigation of a promising approach of leveraging powerful techniques from machine learning to implement intelligent sensor networks. There are many interesting directions warranting further research. For instance, for large-scale networks with mobile readers (e.g., for Smart Cities), machine learning can be applied to jointly design reader inference and path planning. As another example, inference from large-scale sensing data can be centralized  in the cloud using advanced machine learning such as deep learning. Alternatively, a hierarchical framework can be developed for data inference combining centralized and distributed inference. All such interesting directions call for novel designs integrating  machine learning and wireless communications in the same vein as the current work but in potentially more complex ways.  

\appendix

\subsection{Proof of Lemma \ref{prop:1}}\label{appendix:prop:1}
Starting from the cumulative distribution function (CDF) of $\boldsymbol \theta$, denoted by $F_{\boldsymbol \theta}(\bu)$, by definition and the mapping function shown in \eqref{mapping} we have
\begin{align}\label{CDF of theta}
F_{\boldsymbol \theta}(\bu) 
%&= \text{Prob}(\theta_1 < s_1, \cdots, \theta_N < s_N) \notag\\
& =  \text{Prob}\l(\frac{1}{1+\exp\l(-\frac{x_1-\mu_1}{\sigma_1}\r)} < u_1, \cdots, \frac{1}{1+\exp\l(-\frac{x_N-\mu_N}{\sigma_N}\r)} < u_N\r) \notag\\
& = \text{Prob}\l(x_1 < -\sigma_1 \log\l(\frac{1}{u_1} - 1\r) + \mu_1, \cdots,  x_N < -\sigma_N \log\l(\frac{1}{u_N} - 1\r) + \mu_N \r) \notag \\
& = F_{\bx} \l(-\sigma_1 \log\l(\frac{1}{u_1} - 1\r) + \mu_1, \cdots,  -\sigma_N \log\l(\frac{1}{u_N} - 1\r) + \mu_N \r),
\end{align}
in which $F_{\bx}(\cdot)$ represent the CDF of $\bx$. Then the PDF of $\boldsymbol \theta$ can be evaluated by taking derivative of \eqref{CDF of theta} with respect to $\bu = [u_1, u_2, \cdots, u_N]^T$, namely
\begin{align}\label{PDF of theta}
p(\boldsymbol \theta) &= \left. \frac{\partial }{\partial \bu} F_{\bx} \l(-\sigma_1 \log\l(\frac{1}{u_1} - 1\r) + \mu_1, \cdots,  -\sigma_N \log\l(\frac{1}{u_N} - 1\r) + \mu_N \r) \right|_{\bu = \boldsymbol \theta} \notag \\
& = \prod_{n=1}^N \frac{\sigma_n}{\theta_n - \theta_n^2} f_{\bx}\l(-\sigma_1 \log\l(\frac{1}{\theta_1} - 1\r) + \mu_1, \cdots,  -\sigma_N \log\l(\frac{1}{\theta_N} - 1\r) + \mu_N \r),
\end{align}
where $f_{\bx}(\cdot) = p(\bx)$ denotes the PDF of $\bx$ as given by \eqref{PDF of x}. 
%For ease of notation, let $p(\boldsymbol \theta) = f_{\boldsymbol \theta}(\bu)|_{\bu = \boldsymbol \theta}$. 
Then the desired result can be obtained by substituting \eqref{PDF of x} into \eqref{PDF of theta}, completing the proof. 

\subsection{Proof of Proposition \ref{prop:3}}\label{appendix:prop:3}
Let $\tilde{\cal L}(\theta)$ denote the objective function shown in \eqref{Optimization for updating theta}, then the first and second derivatives of $\tilde{\cal L}(\theta)$ with respect to $\theta$ can be expressed as follows
\begin{align}
\tilde{\cal L}'(\theta) &= \sum_{i = 1}^L \sum_{m = 0}^N q^{(i)}_T(m) \l( \frac{m}{\theta} - \frac{N-m}{1-\theta} \r). \label{First derivative of L theta}\\
\tilde{\cal L}''(\theta) &= - \sum_{i = 1}^L \sum_{m = 0}^N q^{(i)}_T(m) \l( \frac{m}{\theta^2} + \frac{N-m}{(1-\theta)^2} \r). \label{Second derivative of L theta}
\end{align}

Note that $q^{(i)}_T(m)$ is non-negative as it is a probability measure, hence it is easy to see that $\tilde{\cal L}''(\theta) < 0,\;\; \forall \theta \in (0,1)$, suggesting $\tilde{\cal L}(\theta)$ is a concave function over its domain $\theta \in (0,1)$. As a result, the maximum of $\tilde{\cal L}(\theta)$ can be achieved by setting its first derivative to zero, leading to the desired result in \eqref{update for theta}. However, one should take note that the zero of $\tilde{\cal L}'(\theta)$ may not lie in the range of $(0,1)$, thereby requiring further validation of the solution. 
Fortunately, the result shown in \eqref{update for theta} is guaranteed to meet the constraint $\theta \in (0,1)$. To see this, 
let's rewrite \eqref{update for theta} as 
$\theta^* = \frac{1}{L}\sum_{i = 1}^L \l(\frac{1}{N}\sum_{m = 0}^N q^{(i)}_T(m) m \r)$.
Note that the term inside the bracket can be identified as the expectation of $T^{(i)}$ with respect to its posterior distribution $q^{(i)}$ normalized by its maximum value $N$ which should fall in the range $(0,1)$ for each $i$. Hence a further average over $i$ still keeps the value within $(0,1)$, which completes the proof. 

%\section{Proof of Proposition \ref{prop:4}}\label{appendix:prop:4}
\subsection{Proof of Proposition \ref{prop:4}}\label{appendix:prop:4}
Denoting $\tilde{\cal L}(\Sigma_m)$ as the objective function shown in \eqref{Optimization for updating Sigma_t}. Then taking the first and second derivatives of $\tilde{\cal L}(\Sigma_m)$ with respect to $\Sigma_m$ gives
\begin{align}
\nabla_{\Sigma_m} \tilde{\cal L}(\Sigma_m) = \sum_{i = 1}^L  q^{(i)}_T(m) \l( -\frac{M}{\Sigma_m} - \frac{\|\by^{(i)}\|^2}{2\Sigma_m^2} \r) \label{First derivative of L Sigma_t}.
\end{align}
\begin{align}
\nabla^2_{\Sigma_m} \tilde{\cal L} (\Sigma_m) = - \sum_{i = 1}^L  q^{(i)}_T(m) \l( \frac{M\Sigma_m - \|\by^{(i)}\|^2}{\Sigma_m^3} \r) = \frac{\Sigma_m - \frac{\sum_{i = 1}^L  q^{(i)}_T(m) \| {\by^{(i)}}\|_2^2}{M\sum_{i = 1}^L  q^{(i)}_T(m)}}{\Sigma_m^3/(M\sum_{i = 1}^L  q^{(i)}_T(m))}. \label{Second derivative of L Sigma_t 2}
\end{align}
Although $\nabla^2_{\Sigma_m} \tilde{\cal L} (\Sigma_m)$ is not non-positive for all $\Sigma_m \in \mathbb{R^+}$, in other words,  $\tilde{\cal L}(\Sigma_m)$ is not a concave function on its domain $\Sigma_m \in \mathbb{R}$, we can still prove that the maximum of $\tilde{\cal L}(\Sigma_m)$ is attained at the zero of its first derivative due to its quasi-concavity as shown below. Particularly, note that there is only one root for $\nabla_{\Sigma_m} \tilde{\cal L}(\Sigma_m) = 0$, namely $\Sigma_m^*$ derived in \eqref{update for Sigma_t}. Furthermore, as can be seen from \eqref{Second derivative of L Sigma_t 2} that the second derivative at that critical point is strictly less than zero, i.e., $\nabla^2_{\Sigma_m} \tilde{\cal L} (\Sigma_m^*) < 0$. The above observations plus the fact that $\tilde{\cal L}(\Sigma_m)$ is a continuous function and differentiable everywhere on its domain $\Sigma_m > 0$  lead to the desired result.

\bibliographystyle{ieeetran}
\bibliography{IEEEabrv,BackSensing_Related_Works}

%\newpage 
%
%\hl{XXX: 1) Italic for full names of all abbreviations. 2) The PDF notation is confusing: sometimes $p(x)$ sometimes $f_x(r)$. Unify the notation. }

\end{document}

%% file: header.tex
\newtheorem{lemma}{Lemma}

\newtheorem{remark}{Remark}

\newtheorem{assumption}{Assumption}

\def\l{\left}
\def\r{\right}
\def\({\left(}
\def\){\right)}

\def\br{{\mathbf{r}}}
\def\bs{{\mathbf{s}}}

\def\bu{{\mathbf{u}}}

\def\bx{{\mathbf{x}}}
\def\by{{\mathbf{y}}}
\def\bz{{\mathbf{z}}}
\def\b0{{\mathbf{0}}}

\def\bC{{\mathbf{C}}}
\def\bD{{\mathbf{D}}}

\def\bY{{\mathbf{Y}}}
\def\bZ{{\mathbf{Z}}}

\newcommand{\nn}{\nonumber}